\documentclass[11pt]{article}
\usepackage{preamble}

\usepackage[margin=1in]{geometry}
\title{Quantum Probe Tomography}
\mathtoolsset{showonlyrefs}

\allowdisplaybreaks

\usepackage[lined,boxed,ruled,norelsize,algo2e,linesnumbered]{algorithm2e}

\usepackage{defs}

\usepackage{textgreek}

\author{
Sitan Chen\thanks{Joint first author} \\
Harvard University \\
\href{mailto:sitan@seas.harvard.edu}{\texttt{sitan@seas.harvard.edu}}
\and
Jordan Cotler$^*$ \\
Harvard University \\
\href{mailto:jcotler@fas.harvard.edu}{\texttt{jcotler@fas.harvard.edu}}
\and
Hsin-Yuan Huang \\
Caltech \\
\href{mailto:hsinyuan@caltech.edu}{\texttt{hsinyuan@caltech.edu}}
}

\newcommand{\params}{{\boldsymbol{\lambda}}}
\newcommand{\vDelta}{\boldsymbol{\Delta}}
\newcommand{\vdelta}{\boldsymbol{\delta}}
\newcommand{\vphi}{{\boldsymbol{\phi}}}

\newcommand{\proj}{\mathrm{proj}}
\newcommand{\Hess}{\operatorname{\textsf{H}}}

\newcommand{\maxj}{{\overline{j}}}
\newcommand{\maxk}{{\overline{k}}}

\newcommand{\Od}{\mathcal{O}_{\mathfrak{d}}}

\newcommand{\mumax}{R}
\newcommand{\lmax}{\Lambda}

\newcommand{\bx}{\mathbf{x}}
\newcommand{\bc}{\mathbf{c}}
\newcommand{\by}{\mathbf{y}}

\newcommand{\hideg}[3]{\mathsf{D}^{#1:[#2,#3]}_{\mu,C}(\beta)}

\newcommand{\vmu}{\boldsymbol{\mu}}
\newcommand{\smoothness}{\upsilon}
\newcommand{\smoothing}{\text{\textvarsigma}}
\newcommand{\curvature}{\underline{\sigma}}
\renewcommand{\norm}[1]{\|#1\|}

\renewcommand{\top}{\intercal}

\newcommand{\obsvalue}[4]{A_{#3,#4}(#2;#1)}
\renewcommand{\i}{\mathrm{i}}

\newcommand{\newdelt}[4]{\Delta^{(#1)}_{#3,#4}(#2)}
\newcommand{\hatdelt}[4]{\widehat{\Delta}^{(#1)}_{#3,#4}(#2)}

\newcommand{\jac}{\mathrm{jac}}

\renewcommand{\Var}{\mathrm{Var}}

\newcommand{\zeroset}{\mathbb{V}}
\renewcommand{\C}{\mathbb{C}}

\newcommand{\probeindex}{\mathbf{0}}

\usepackage[titles]{tocloft}

\usepackage{stackengine}
\stackMath
\newcommand\tsup[2][2]{%
 \def\useanchorwidth{T}%
  \ifnum#1>1%
    \stackon[-.5pt]{\tsup[\numexpr#1-1\relax]{#2}}{\scriptscriptstyle\sim}%
  \else%
    \stackon[.5pt]{#2}{\scriptscriptstyle\sim}%
  \fi%
}

\begin{document}
\pagestyle{empty}
{
  \renewcommand{\thispagestyle}[1]{}
  \maketitle

\begin{abstract}
Characterizing quantum many-body systems is a fundamental problem across physics, chemistry, and materials science. While significant progress has been made, many existing Hamiltonian learning protocols demand digital quantum control over the entire system, creating a disconnect from many real-world settings that provide access only through small, local probes. Motivated by this, we introduce and formalize the problem of \emph{quantum probe tomography}, where one seeks to learn the parameters of a many-body Hamiltonian using a single local probe access to a small subsystem of a many-body thermal state undergoing time evolution. We address the identifiability problem of determining which Hamiltonians can be distinguished from probe data through a new combination of tools from algebraic geometry and smoothed analysis. Using this approach, we prove that generic Hamiltonians in various physically natural families are identifiable up to simple, unavoidable structural symmetries. Building on these insights, we design the first efficient end-to-end algorithm for probe tomography that learns Hamiltonian parameters to accuracy $\varepsilon$, with query complexity scaling polynomially in $1/\varepsilon$ and classical post-processing time scaling polylogarithmically in $1/\varepsilon$. In particular, we demonstrate that translation- and rotation-invariant nearest-neighbor Hamiltonians on square lattices in one, two, and three dimensions can be efficiently reconstructed from single-site probes of the Gibbs state, up to inversion symmetry about the probed site. Our results demonstrate that robust Hamiltonian learning remains achievable even under severely constrained experimental access.
\end{abstract}

}

\clearpage
\pagestyle{plain}
% \clearpage
\pagenumbering{arabic}

\tableofcontents

\newpage

\section{Introduction}
\label{sec:intro}

Understanding the properties of complex quantum many-body systems remains one of the central challenges in condensed matter physics, quantum chemistry, materials science, and quantum information science. From quantum sensing and metrology \cite{de2005quantum,valencia2004distant,leibfried2004toward,bollinger1996optimal,lee2002quantum,mckenzie2002experimental,holland1993interferometric,wineland1992spin,caves1981quantum} to quantum device engineering \cite{boulant2003,innocenti2020,ben2020,shulman2014,sheldon2016,sundaresan2020} and quantum many-body physics \cite{wiebe2014a,wiebe2014b,verdon2019,burgarth2017, wang2017, kwon2020, wang2020, cotler2020quantum, huang2020predicting}, the ability to efficiently extract information about unknown quantum Hamiltonians is crucial for advancing our fundamental understanding and technological capabilities. This challenge has motivated extensive research in Hamiltonian learning \cite{li2020hamiltonian,che2021learning,yu2022,hangleiter2021,FrancaMarkovichEtAl2022efficient,ZubidaYitzhakiEtAl2021optimal,BaireyAradEtAl2019learning,GranadeFerrieWiebeCory2012robust,gu2022practical,wilde2022learnH,KrastanovZhouEtAl2019stochastic, huang2022learning, gu2024practical, bakshi2024structure, haah2024learning, bakshi2024learning}, where the goal is to reconstruct an unknown Hamiltonian $H$ from experimental observations of a physical system governed by $H$.

Recent theoretical advances in Hamiltonian learning have achieved remarkable progress, including algorithms that attain the Heisenberg-limited scaling of $\mathcal{O}(\epsilon^{-1})$ total evolution time for learning many-body Hamiltonians to precision $\epsilon$ \cite{huang2022learning, gu2024practical, bakshi2024structure} and polynomial-time algorithms for learning $H$ from its Gibbs states \cite{haah2024learning, bakshi2024learning, chen2025learning}. These algorithms provide fundamental insights into the theoretical limits of quantum system characterization. However, they rely on experimental capabilities that may not be readily available in practice: the ability to measure and perform quantum gate operations on all $n$ qubits, and the capability to seamlessly transition between digital quantum circuit operations and analog evolution under the unknown Hamiltonian. Implementing these digital quantum circuit operations requires having already characterized and calibrated the quantum device to high precision. This creates a logical circularity, as quantum device characterization is precisely what we seek to accomplish.

The gap between theoretical capabilities and experimental reality is particularly pronounced when studying naturally occurring quantum many-body systems. Consider characterizing a novel quantum material, elucidating molecular dynamics, or probing correlations in ultracold atoms. In these settings, experimentalists typically lack individual control over microscopic degrees of freedom and must work with local probes that access only small regions of the system. The assumption of flexible, system-wide control inherent in existing algorithms renders them inapplicable to many of the most scientifically important scenarios.

This disconnect motivates the following fundamental question towards bridging the gap between our theoretical models of learning and experimental realities:
\begin{center}
\emph{How can we learn about a large quantum many-body system using only a small quantum probe?}    
\end{center}
This is a problem of both experimental and theoretical importance, and has been studied in a number of works~\cite{burgarth2009coupling, burgarth2009indirect, di2009hamiltonian, zhang2014quantum, zhang2015identification, sone2017hamiltonian, sone2017exact, che2021learning, schuster2023learning}.  Indeed, consider the experimental scenario where an experimentalist has access to a small, highly controllable quantum probe that can be brought into interaction with only a tiny fraction of a much larger quantum system.
By controlling the quantum probe and the probe-system interactions, the experimentalist could effectively access a small region of the large system.
This setting captures the essence of many modern quantum experiments: a scanning probe microscope tip interacting with a material surface~\cite{gring2012relaxation, cheuk2016observation, gross2021quantum}, a single trapped ion brought near an optical lattice~\cite{kotler2011single, bakr2009quantum}, a nitrogen-vacancy center in diamond sensing its local magnetic environment~\cite{taylor2008high, maze2008nanoscale, balasubramanian2008nanoscale, maletinsky2012robust, boss2016one}, or a controllable ancilla qubit that can couple to atoms in an optical lattice~\cite{micheli2004single, weitenberg2011single, fukuhara2013quantum, elliott2016nondestructive}.

\begin{figure}[t]
    \centering
    \includegraphics[width=0.95\linewidth]{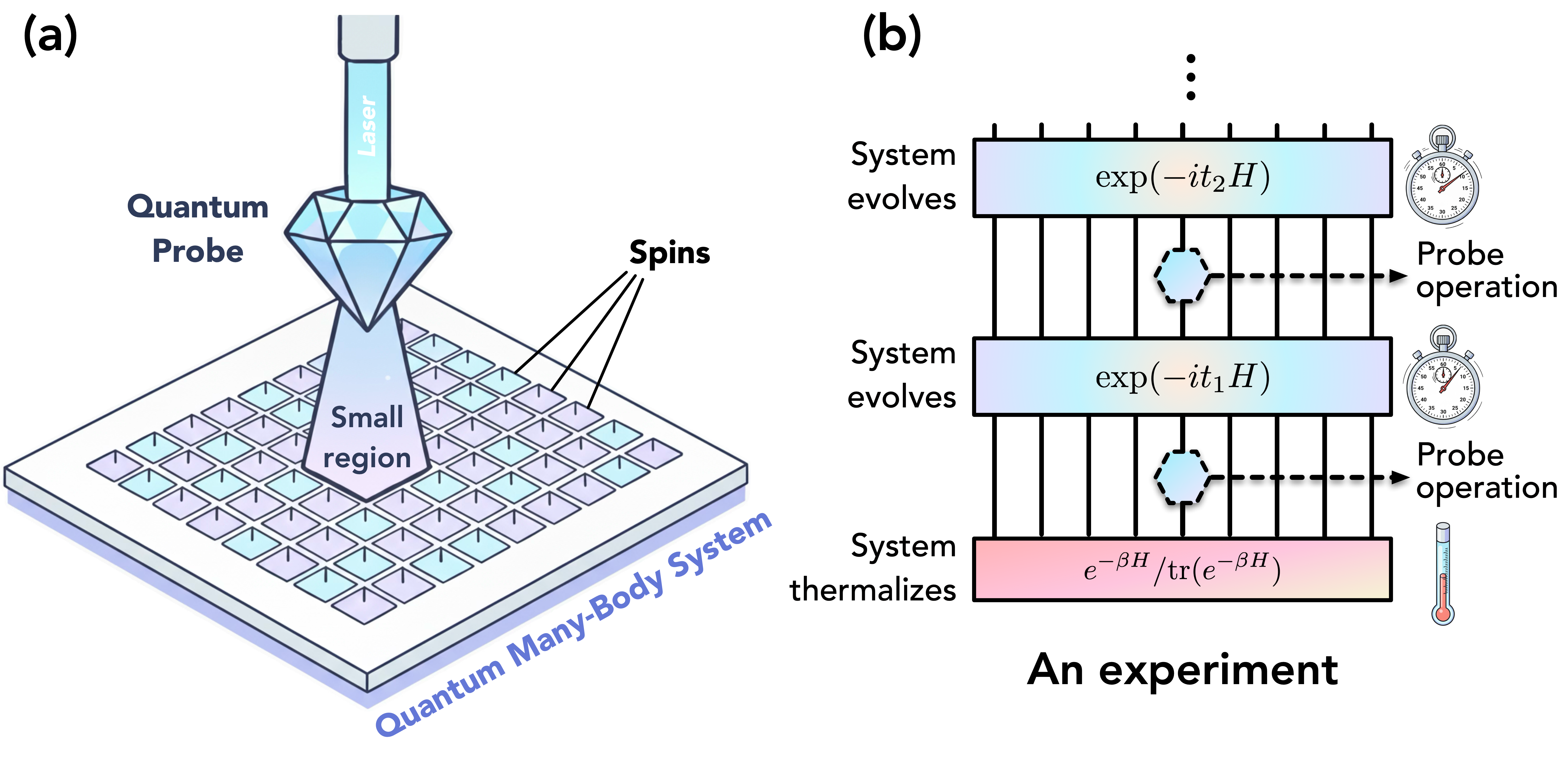}
    \caption{Cartoon depiction of \emph{quantum probe tomography}. The goal is to learn about a quantum many-body system using a small quantum probe. \textbf{(a)} An illustration of an experimental platform for learning a quantum many-body system (depicted as a two-dimensional array of spins) using a small quantum probe (depicted as a diamond). The probe can only observe and affect a small region of the much larger quantum many-body system. \textbf{(b)} In the model of computation, we consider an experiment to consist of three types of steps: 1. the many-body system thermalizes to an inverse temperature $\beta$ resulting in the thermal state $\frac{e^{-\beta H}}{\tr(e^{-\beta H})}$, 2. the quantum probe performs a quantum channel, a quantum measurement, or a combination of both (quantum instrument) to a small region of the spins, 3. the system evolves under the native Hamiltonian $H$. }
    \label{fig:QuantumProbe}
\end{figure}

In this work, we synthesize previous approaches and formalize this conceptual question into a learning-theoretic problem that we term \emph{quantum probe tomography}; see Figure~\ref{fig:QuantumProbe} for an illustration. Here, an experimentalist seeks to learn a many-body quantum Hamiltonian~$H$ using only severely constrained access to the many-body system via a small quantum probe (or a collection of probes). For theoretical tractability, we assume that high controllability of the probe and carefully engineered probe-system interactions enable precise control over a single site within the many-body system. This abstraction captures the essential limitation of constrained access while avoiding the complexities of modeling detailed probe-system interactions. As a simple example of the problem class, we consider the experimental platform to offer the following operations:
\begin{itemize}
    \item The system thermalizes to the Gibbs state $\rho_\beta \propto e^{-\beta H}$ for tunable inverse temperature $\beta$.
    \item The system evolves under the Hamiltonian dynamics $e^{-\i Ht}$ for controllable time intervals $t$.
    \item The experimentalist can apply only single-qubit operations to the central site of the system.
\end{itemize}
Of course there are many generalizations of the above that would still constitute quantum probe tomography, such as if the probe touches more than one degree of freedom, if there are multiple probes at different locations, if the state is pure instead of a Gibbs state, etc.  We will focus on the above problem setting in this paper for concreteness.

\emph{A priori}, it is not clear which aspects of the Hamiltonian $H$ can be extracted using the severely constrained single-qubit access to the many-body system described above. The fundamental challenge lies in the \emph{identifiability problem} (see e.g.~\cite{sone2017hamiltonian}), which asks whether the limited local information accessible through single-qubit operations can reveal the many-body interactions throughout the lattice. This challenge is particularly acute because local measurements can be insufficient to distinguish between many continuous or discrete sets of distinct Hamiltonians if they produce identical statistics at the probe site.

In this work, we resolve the identifiability problem in generic settings by combining tools from algebraic geometry with smoothed analysis techniques. Our approach yields the first efficient end-to-end learning algorithm capable of Hamiltonian reconstruction using only single-site probe access. For the concrete, physically-relevant case of translation- and rotation-invariant nearest-neighbor Hamiltonians in one-, two-, three-dimensional lattices, we establish the following result: given any Hamiltonian $H_{\mathrm{hard}}$ that may exhibit arbitrarily hard identifiability issues, when we consider a generic Hamiltonian $H$ drawn from any small neighborhood of $H_{\mathrm{hard}}$, we can efficiently learn $H$ to high precision, up to simple and unavoidable symmetries.

\section{Results}
\label{sec:results}

\subsection{Defining Quantum Probe Tomography}

We introduce a mathematical framework for modeling the real-world situation where an observer seeks to characterize an unknown quantum many-body system using only small local quantum probes and limited control capabilities.

\subsubsection{Computational Model}

Consider a quantum many-body system governed by an unknown Hamiltonian $H$ acting on a Hilbert space $\mathcal{H}$. The system may represent spin lattices, quantum gases, solid-state systems, or molecular complexes. Each experiment follows this structure:
\begin{enumerate}
    \item \emph{Initial State Preparation:} The state of the system thermalizes to 
    \begin{equation}
    \rho_\beta = \frac{e^{-\beta H}}{\tr[e^{-\beta H}]}
    \end{equation}
    for a tunable inverse temperature $\beta \geq 0$.
    
    \item \emph{Time-dependent Probing:} For $j$ from $1$ to $J$:
    \begin{itemize}
        \item \emph{Probe Operation:} The observer probes a small subsystem by performing a quantum instrument $\{\mathcal{E}^{(j)}_{x_j}\}_{x_j}$ on a small subset $S_j$ of the total system. This yields a classical outcome $x_j$ with probability $p(x_j) = \tr[\mathcal{E}^{(j)}_{x_j}(\rho_{j-1})]$ and evolves the state to
        \begin{equation}
        \rho'_j = \frac{\mathcal{E}^{(j)}_{x_j}(\rho_{j-1})}{\tr[\mathcal{E}^{(j)}_{x_j}(\rho_{j-1})]},
        \end{equation}
        where $\rho_0 = \rho_\beta$ is the initial state.
        \item \emph{Hamiltonian Evolution:} The observer waits for time $t_j$, causing the state to evolve under the unknown Hamiltonian $H$:
        \begin{equation}
        \rho_j = e^{-\i t_j H} \rho'_j e^{\i t_j H}.
        \end{equation}
        Then we move on to $j \leftarrow j + 1$.
    \end{itemize}
\end{enumerate}
The experimental data obtained from one experiment consists of the classical measurement outcome sequence $\mathbf{x} = (x_1, x_2, \ldots, x_J)$ from the probe operations.

We will primarily focus on the case where the observer has a single probe at the probe site indexed by $\probeindex$. This corresponds to the chosen subsystem $S_j = \{\probeindex\}$ for all $j$. Furthermore, instead of general quantum instruments, we will focus solely on performing single-qubit channels and ending with a projective measurement on the probe qubit at site $\probeindex$.

\subsubsection{Physical Rationale}

We explain the rationale behind various aspects of this mathematical model in the following.

\paragraph{Quantum Instruments for Probe Operations} The quantum instrument formalism $\{\mathcal{E}^{(j)}_{x_j}\}_{x_j}$ captures the dual nature of experimental probes: they extract classical information through measurement outcomes $x_j$ while simultaneously disturbing the quantum system, evolving it from $\rho_{j-1}$ to $\rho'_j$ in a way that depends on both the outcome and the probe's physical interaction. This mathematical description encompasses diverse scenarios including scanning probe microscopy (simultaneous measurement and sample perturbation), trapped ion probing (spin correlation detection with local rotations), and NV center magnetometry (field measurement with microwave-driven dynamics).

\paragraph{Constrained Access} The restriction to small subsystems $S_j$ reflects fundamental experimental limitations in probing quantum many-body systems, where researchers cannot individually control every microscopic degree of freedom but must work through local probes with limited spatial resolution. This constraint distinguishes quantum probe tomography from idealized quantum process tomography, which assumes arbitrary global control. Probe configurations range from single-site probing ($|S_j| = 1$) to multi-site probing ($|S_j| = k$ for small $k$), multiple simultaneous probes on disjoint subsystems ($S_j$ is not a single contiguous region), and mobile probes that can be repositioned between measurements ($S_j$ changes throughout the experiments). 

\paragraph{Interleaved Dynamics} The alternation between probe operations and free Hamiltonian evolution $e^{-it_j H}$ provides a computationally tractable approximation to realistic experimental protocols. In actual experiments, probe operations and system evolution occur simultaneously rather than sequentially as it is physically impossible to completely halt the system's intrinsic dynamics. The most physically realistic description would involve continuous evolution under a time-dependent Hamiltonian $H(t)$ that includes both the system's intrinsic dynamics $H$ and the probe's interaction, with continuously monitored weak measurements that yield infinitesimal information while causing infinitesimal state perturbations. The continuous process emerges by taking the following limit in the discrete time evolution: as the number of probe operations $J \to \infty$, the evolution times $t_j \to 0$, and each quantum instrument becomes a weak measurement with infinitesimal disturbance, the discrete alternation approaches the true continuous dynamics.

\vspace{1em}
The objectives of the learning problem could be full Hamiltonian reconstruction to precision $\epsilon$, property estimation (ground state energies, spectral gaps, correlation functions), quantum phase classification, and model selection between competing theoretical descriptions. Experimental constraints can be naturally embedded to this model through limited probe access ($|S_j|$ is small), restricted operation sets (such as Pauli-only measurements), finite temperature effects, decoherence, and bounded control over evolution times and temperatures. These constraints define the practical boundaries within which quantum probe tomography algorithms must operate.

\subsubsection{Probing Translation-Invariant Hamiltonians}

To illustrate the framework in a practical scenario, we focus on an important class of nearest-neighbor, translation-invariant, rotation-invariant Hamiltonians on $D$-dimensional square lattices. This class encompasses many physically relevant systems including quantum magnets, ultracold atomic gases in optical lattices, and fundamental condensed matter models. Such systems are described by the following Hamiltonian,
\begin{equation}
\label{E:NNH1}
H = \sum_{\langle v, v'\rangle} \sum_{\mu,\nu=1}^3 J_{\mu \nu} \,\sigma^\mu_v \sigma^\nu_{v'} + \sum_{v} \sum_{\mu=1}^3 h_\mu \sigma^\mu_v,
\end{equation}
where $v$ labels each site of the $D$-dimensional square lattice, $\langle v, v'\rangle$ denotes nearest-neighbor pairs, and $\sigma^\mu_v$ are Pauli operators ($\mu \in \{1,2,3\}$ corresponding to $X$, $Y$, $Z$) acting on site $v$. The coupling matrix $J = \{J_{\mu \nu}\}_{\mu, \nu=1,2,3}$ characterizes the nearest-neighbor interactions, while the local field vector $\mathbf{h} = (h_1, h_2, h_3)$ determines the single-site terms. Translation invariance and rotation invariance ensure that these parameters are identical across all lattice sites and bonds.

This parameterization captures essential many-body physics while remaining tractable for theoretical analysis. The coupling matrix $J$ encodes various interaction types: diagonal terms $J_{\mu \mu}$ represent Ising-type interactions along different axes, while off-diagonal terms $J_{\mu \nu}$ with $\mu \neq \nu$ generate XY-type exchanges and Dzyaloshinskii-Moriya interactions. The local fields $h_\mu$ break spin symmetries and can drive quantum phase transitions.  Despite having only $12$ uniform nearest-neighbor parameters, this translation- and lattice-rotation–invariant spin-$\frac{1}{2}$ model already realizes key magnetic phases, including ferromagnetic and N\'{e}el antiferromagnetic order~\cite{manousakis1991spin,auerbach2012interacting}, as well as canted and chiral states arising from Dzyaloshinskii–Moriya and related off-diagonal anisotropies~\cite{dzyaloshinsky1958thermodynamic,moriya1960anisotropic,shekhtman1992moriya}. Mild extensions such as bond-dependent anisotropy, geometric frustration, or further-neighbor interactions, can additionally stabilize quantum spin-liquid and topologically ordered phases~\cite{kitaev2006anyons,balents2010spin,savary2016quantum,broholm2020quantum,yan2011spin,wen2017colloquium,motrunich2005variational}. We further impose the constraint that the observer can only access a single site $S_j = \{ \probeindex \}$, which we consider to be at the center of the square lattice.

The learning objective becomes concrete: given experimental access through quantum probe tomography, estimate the coupling matrix $J$ and field vector $\mathbf{h}$ to precision $\epsilon$. This reduces the infinite-dimensional Hamiltonian learning problem to estimating a finite parameter vector in $\mathbb{R}^{12}$, making the analysis mathematically tractable while preserving the essential challenge of extracting global many-body information from local measurements.

\subsubsection{Identifiability Problem}

Given a quantum probe tomography protocol of the type specified above, a central question arises: which sets of Hamiltonians are indistinguishable under the accessible measurements? This is the \emph{identifiability problem}, which has been studied in various forms throughout the literature \cite{sone2017hamiltonian, nielsen2021gate, van2022hardware, huang2022foundations, chen2023learnability}. Understanding the structure of indistinguishable Hamiltonians is crucial for determining what information can realistically be extracted from constrained experimental access.

To illustrate the core challenge, consider the class of 1D, nearest-neighbor, translation-invariant Hamiltonians, where we only have access to probe the central site labeled by $\probeindex$. One can verify that protocols of the type outlined above cannot distinguish two Hamiltonians $H$ and $H'$ if they are related by an invertible linear transformation
\begin{align}
H' = (L^{-1} \otimes \mathds{1}_\probeindex) \,H\, (L \otimes \mathds{1}_\probeindex),
\end{align}
where $L$ acts on the non-probed sites and $\mathds{1}_\probeindex$ is the identity on the probed site. Other transformations, including nonlinear ones, may exist depending on the specific quantum probe protocol, but we focus on this linear setting for concreteness.

The natural question becomes: do there exist distinct Hamiltonians $H$ and $H'$ related as above, such that both belong to the promised class of 1D, nearest-neighbor, translation-invariant Hamiltonians? The answer is \emph{yes}. Consider the quantum Ising model at its critical point
\begin{align}
H = \sum_{i \in \mathbb{Z}} \left(Z_i Z_{i+1} + X_i\right).
\end{align}
Let $L \otimes \mathds{1}_{\probeindex} = \cdots \mathds{1} \otimes X \otimes \mathds{1} \otimes X \otimes \mathds{1} \otimes \cdots$. Then
\begin{align}
\label{E:IsingIdentify1}
H' = (L^{-1} \otimes \mathds{1}_{\probeindex}) \,H\, (L \otimes \mathds{1}_\probeindex) = \sum_{i \in \mathbb{Z}} \left(-Z_i Z_{i+1} + X_i\right).
\end{align}
Therefore, the ferromagnetic and antiferromagnetic critical points are fundamentally indistinguishable under any number of experiments conducted using the single-site probe.

More generally, given a promised class of Hamiltonians, we can ask: which sets of Hamiltonians are indistinguishable under the accessible measurements? In principle, for a given access model, one can attempt to classify the complete equivalence classes of indistinguishable Hamiltonians \cite{van2022hardware}. However, performing such a comprehensive classification, even for relatively simple access models, appears to be an extremely challenging problem.

A more tractable approach is to address the identifiability problem when we focus instead on \emph{generic} Hamiltonians within the promised class through the lens of \emph{smoothed analysis}. This perspective, originating in theoretical computer science \cite{carbery2001distributional, spielman2004smoothed, arthur2009k}, examines whether hard worst-case instances become tractable when the worst-case input is slightly perturbed. The key intuition is that pathological cases requiring intricate parameter fine-tuning often become resolvable when we introduce small random perturbations. This smoothed analysis perspective differs from \emph{average-case complexity} analysis. The latter studies the expected difficulty over randomly chosen instances from a specific distribution (e.g., standard Gaussian), whereas smoothed analysis quantifies over all distributions given by random perturbations starting from an arbitrary, possibly worst-case hard input.

Concretely, we consider the following question: suppose we start with any Hamiltonian $H$ from our promised class, and we add a small random perturbation $\delta H$ where $\|\delta H\| \leq \sigma$ for some noise scale $\sigma > 0$. Does the perturbed Hamiltonian $H + \delta H$ remain indistinguishable from some other Hamiltonian in the class, or does the perturbation generically circumvent the indistinguishability? This approach is especially relevant for quantum probe tomography because realistic physical systems always contain some degree of disorder, imperfections, or environmental coupling that effectively provide natural perturbations.

We find, somewhat remarkably, that solving the identifiability problem for slightly perturbed instances is tractable using a combination of tools from algebraic geometry, symmetry analysis, and anti-concentration. Moreover, our findings are encouraging: in practical cases of quantum probe tomography, we can completely learn generic Hamiltonians up to mild symmetries, meaning that generic Hamiltonians are essentially identifiable under realistic experimental conditions.

As a concrete example, consider the class of nearest-neighbor, translation-invariant, rotation-invariant Hamiltonians specified in Eq.~\eqref{E:NNH1}. If we perform quantum probe tomography on a thermal state at a single site, then we cannot distinguish two Hamiltonians if they are related by inversion $\textsf{P}$ about that site, namely
\begin{equation}
H' = \textsf{P} H \textsf{P} = \sum_{\langle v, v'\rangle} \sum_{\mu, \nu=1}^3 J_{\nu \mu}\, \sigma^\mu_v \sigma^\nu_{v'} + \sum_{v} \sum_{\mu=1}^3 h_\mu \sigma^\mu_v,
\end{equation}
which simply exchanges $J_{\mu \nu}$ with $J_{\nu \mu}$. We demonstrate that for any Hamiltonian in this family, we can fully learn its parameters up to this single unavoidable inversion symmetry $\textsf{P}$. Furthermore, if we extend the protocol to probe two sites instead of one, then even this inversion indistinguishability disappears for generic instances, enabling complete parameter recovery.

\subsection{End-to-End Learning Guarantee}

Our first main result provides a general template for algorithmically solving the identifiability problem, stipulating a series of (efficiently certifiable) non-degeneracy conditions under which one can recover the true parameters of a generic Hamiltonian up to simple and unavoidable symmetries given probe access. Before stating the result, we describe the algorithm at a high level, which operates in three stages:
\begin{enumerate}[leftmargin=*]
    \item \textit{Estimate observable derivatives}. We will work with simple experiments of the following form. Starting with $\rho_\beta$, perform a single-qubit channel $C$ at the probed site $\probeindex$, undergo Hamiltonian evolution via $H$ for time $t$, and then measure the qubit at $\probeindex$ in Pauli-$\mu$ basis for some $\mu\in\{X,Y,Z\}$. The resulting observable value $\obsvalue{\beta}{t}{\mu}{C}$ is a nonlinear function in the parameters of the unknown Hamiltonian. By performing finite differencing over many such experiments, we can approximately extract derivatives of this observable with respect to $\beta, t$, which are \emph{polynomial} functions in the parameters of $H$.
    \item \textit{Approximately solve associated polynomial system}. The experimental data from Stage 1 thus naturally gives rise to a system of polynomial constraints. Because we are working with translation-invariant Hamiltonians, the number of variables is $\mathcal{O}(1)$, so the key difficulty lies not in finding a solution (for which simple exhaustive enumeration suffices), but in proving that such a solution must be close up to symmetries to the true parameters. We discuss this in greater detail below.
    \item \textit{Refine iteratively via Newton method}: To achieve better dependence on the target accuracy $\epsilon$, once we have obtained an approximate solution to the polynomial system to within some modest accuracy, we perform a sequence of Newton steps; see Section~\ref{sec:local}. If the approximate solution was already somewhat close to the true parameters, the Newton steps will converge exponentially quickly to the true parameters.
\end{enumerate}

\begin{theorem}[Informal, see Theorem~\ref{thm:main_general}]\label{thm:main_informal}
    Let $P(\bx) \approx \tilde{\bc}$ denote the approximate polynomial system obtained from estimating certain derivatives of $\obsvalue{\beta}{t}{\mu}{C}$. If $H$ is a generic member (in the smoothed analysis sense) of a Hamiltonian family for which $P(\bx)$ satisfies certain efficiently certifiable non-degeneracy assumptions, then there is an algorithm {\sc ProbeLearn} which, given probe access to $H$ at a single qubit, finds Hamiltonian parameters which are $\epsilon$-close to the parameters of $H$, up to simple, unavoidable symmetries, in query complexity $\mathrm{poly}(1/\epsilon)$ and classical post-processing time $\mathcal{O}(\log^2(1/\epsilon))$ with high probability.
\end{theorem}

\noindent While the general template bears resemblance to existing methods for Hamiltonian learning, we stress that the key difficulty and novelty lies not in the algorithm design, but in the \emph{analysis}, in particular of Stage 2. As it turns out, for the polynomial systems we consider, even at high temperatures and short times, it is not at all clear that the experimental data obtained in Stage~1 is sufficient to recover $H$, or how to ascertain this. In Section~\ref{sec:examples}, we explicitly compute the relevant polynomial system in the translation-invariant setting and note that \emph{a priori} one needs a computer algebra system to verify for any specific problem instance whether the system uniquely identifies the true parameters up to simple and unavoidable symmetries. Worse yet, as the quantum Ising model example above demonstrates, the polynomial system arising for non-generic problem instances might have many more roots than just the ``correct'' ones. All of these issues lie in stark contrast to, e.g., Gibbs state learning where at sufficiently high temperature, one can read off the terms in $H$ simply from the first-order term in the expansion $\tr(\sigma_a \rho_\beta) = -\frac{\beta}{d}\,\tr(\sigma_a H) + \cdots$.

In this work, we develop a new toolkit marrying techniques from quantitative algebraic geometry and from smoothed analysis to sidestep these issues. Our starting point is the powerful fact that for polynomial systems with non-singular Jacobian at their roots (see Section~\ref{sec:verify-generic-fiber}), whether or not their roots uniquely identify the true Hamiltonian parameters can be verified by checking whether this is the case for a \emph{random} problem instance! Unfortunately, this is not quite enough for our purposes because we never get exact access to the polynomial system, due to perturbations arising from measurement error and finite time and temperature effects. Consequently, the bulk of the work in this paper centers around developing a quantitative analogue of the above theory. We develop a new toolkit for integrating classical ideas from algebraic geometry, specifically from elimination theory, with analytic techniques from smoothed analysis like polynomial anti-concentration. For instance, this toolkit allows us to make statements like: given ``nice'' polynomial maps $F: \C^N\to\C^N$ and $G: \C^N\to \C$, and given an input $\bx^* \in \C^N$ with a small amount of randomness, the values that $G$ takes on over the different solutions $\bx$ to $F(\bx) = F(\bx^*)$ are all quantifiably well-separated from each other. In Section~\ref{sec:quantitative}, we develop ideas like these into a general-purpose guarantee for robust polynomial system solving (Theorem~\ref{thm:solvesystem}) that may be of independent interest and which drives the proof of Theorem~\ref{thm:main_informal}.

Finally, we instantiate this general theory to give the first efficient algorithm for probe tomography for nearest-neighbor, translation-invariant, rotation-invariant Hamiltonians on square lattices:

\begin{theorem}\label{thm:application_intro}
    If $H$ is a generic (in a smoothed analysis sense) nearest-neighbor, translation-invariant, rotation-invariant Hamiltonian on a 1D, 2D, or 3D square lattice, then there is an algorithm with a single-site probe access that makes $\mathrm{poly}(1/\epsilon)$ queries for total evolution time $\mathrm{poly}(1/\epsilon)$ and, after $\mathcal{O}(\log^2 (1/\epsilon))$ classical post-processing time, outputs an $\epsilon$-accurate estimate for the parameters of the Hamiltonian, up to simple, unavoidable symmetries, with high probability.
\end{theorem}

\section{Discussion}
\label{sec:discussion}

In this work, we have formulated a new class of problems in quantum learning theory, namely quantum probe tomography, and developed a suite of quantitative algebraic-geometric tools to address them. We anticipate that these tools will prove broadly useful for tackling other quantum learning problems under restricted control settings. We conclude by highlighting several open problems and promising future directions.

An important direction is to refine our results to achieve better dependence on the error parameter $\varepsilon$ and the smoothing parameter $\sigma$. For the former, one might aim for the optimal Heisenberg scaling $\sim 1/\varepsilon$, while for the latter, a linear dependence $\sim 1/\sigma$ would be preferable to the current polynomial scaling with the power being the number of variables. While we have focused on the physically salient setting of translation-invariant Hamiltonians, which possess only a constant number of parameters, more general settings would benefit from careful optimization of the scaling with respect to the number of parameters. Further improvements in these dependencies could enhance both the theoretical guarantees and practical performance of our algorithms.

Our framework naturally invites extensions to richer physical settings. One compelling direction is quantum probe tomography for systems with long-range interactions, where coupling strengths decay algebraically or exponentially with distance. Another natural generalization involves translation-invariant systems with independent identically distributed on-site disorder, where the goal is to learn both the fixed parameters governing the clean system and the probability distribution characterizing the disordered couplings. Such extensions would capture physically realistic scenarios encountered in many experimental platforms, including trapped ions, Rydberg atoms, and disordered quantum materials. Other natural generalizations include bosonic or fermionic systems.

We emphasize that quantum probe tomography encompasses a rich variety of learning tasks beyond Hamiltonian learning. These include quantum state learning, where the goal is to reconstruct an unknown quantum state through probe interactions, and quantum channel learning, where one seeks to characterize unknown quantum dynamics or noise processes. Each of these tasks presents distinct challenges, and may require tailored applications of our algebraic-geometric methodology. It is our hope that continued development of these tools will lead to experimentally realizable and pragmatic methods for efficient learning of realistic quantum systems.

\subsection*{Acknowledgments}

We would like to thank Anurag Anshu, Daniel Ranard, Josu\'{e} Tonelli-Cueto, and Aditi Venkatesh for valuable conversations. We also thank Weiyuan Gong and Muzhou Ma for comments on the manuscript. SC acknowledges support from NSF CCF-2430375. JC is supported by an Alfred P.~Sloan Foundation Fellowship.  HH~acknowledges support from the Broadcom Innovation Fund, the U.S. Department of Energy, Office of Science, National Quantum Information Science Research Centers, Quantum Systems Accelerator and from the Institute for Quantum Information and Matter, an NSF Physics Frontiers Center.

\appendix

\vspace{2.5em}
\noindent \textbf{\LARGE{}Appendices}
\vspace{0.25em}

\paragraph{Roadmap for Appendices.} In Section~\ref{sec:related} we briefly review the most relevant prior work. In Section~\ref{sec:prelims} we provide technical preliminaries in preparation for the rest of the paper. In Section~\ref{sec:local} we provide a standard analysis of Newton's method which we will use as part of the iterative refinement phase of our learning algorithm. In Section~\ref{sec:estimation}, we prove that derivatives of a certain observable that can be measured in our probe tomography setup can be efficiently estimated from experimental data. In Section~\ref{sec:verify-generic-fiber} we begin developing our algebraic geometric toolkit and give easily checkable conditions under which one can verify the size of generic fibers of a polynomial map. In Section~\ref{sec:quantitative} we make these ideas quantitative and give an algorithmic proof that even if one works with \emph{approximate} polynomial systems, their approximate solutions can be close to the ground truth under suitable non-degeneracy conditions. In Section~\ref{sec:general}, we state our main general result for probe tomography of Hamiltonians for which the associated polynomial system satisfies the desiderata of the previous section, and in Section~\ref{sec:examples}, we instantiate this theory for the specific case of translation-invariant, rotation-invariant, nearest-neighbor Hamiltonians on a $D$-dimensional square lattice for $D = 1,2,3$.

\section{Related Work}
\label{sec:related}

Several prior works have explored learning with quantum probes \cite{burgarth2009coupling, burgarth2009indirect, burgarth2012quantum, di2009hamiltonian, zhang2014quantum, zhang2015identification, sone2017hamiltonian, sone2017exact, che2021learning, schuster2023learning, davis2023probing}. We draw particular attention to \cite{sone2017hamiltonian, wang2020quantum}, which employ the Eigensystem Realization Algorithm (ERA), originally developed for structural system identification in civil engineering, to learn local Hamiltonians from time-trace measurements. These works consider a scenario where one qubit is initialized in a pure state while all other qubits remain in a maximally mixed state. The evolution of that single pure qubit is measured at multiple time points to extract information about the underlying Hamiltonian. The ERA methodology proceeds in two stages: first, raw measurement data is transformed into an approximate numerical transfer function; second, the coefficients of this transfer function are equated with their theoretical symbolic expressions to generate a polynomial system in the Hamiltonian's coupling parameters. These polynomial systems are then solved using Gr\"obner basis algorithms, with numerical validation provided for specific examples.

In \cite{sone2017hamiltonian}, the authors develop a systematic procedure to test Hamiltonian identifiability using Gr\"obner basis techniques. They establish sufficient conditions for identifiability and demonstrate their approach by verifying the identifiability properties for several spin chain models. Building on this foundation, \cite{wang2020quantum} extends the framework using similarity transformation techniques to characterize equivalence classes of indistinguishable Hamiltonians. Both works establish sufficient conditions for identifiability validated on particular instances, with their notion of identifiability generally limited to determining parameter magnitudes, leaving sign ambiguities unresolved.

The ERA-based methodology faces several important limitations that our work addresses. The two-stage pipeline from measurements to transfer functions to polynomial systems lacks guarantees of end-to-end robustness. These works do not establish quantitative bounds on how measurement noise and finite-sample errors propagate through the entire ERA procedure to affect final parameter estimates, nor do they provide explicit convergence rates as a function of target precision $\varepsilon$. Moreover, they do not systematically distinguish between measure-zero pathological cases and generic instances that arise naturally under small perturbations. This leaves open fundamental questions about the scope and robustness of probe tomography: \emph{Is non-identifiability a generic obstacle, or merely an artifact of special parameter configurations?}

Our work introduces a series of mathematical techniques to overcome these limitations. First, we consider a physically more natural setting where the initial state is a thermal state at tunable inverse temperature $\beta$ under the unknown Hamiltonian, rather than requiring the specific preparation where the probe qubit is in a pure state while all other qubits remain in the maximally mixed state. Second, and more critically, we establish a direct, end-to-end connection between experimental observables and Hamiltonian parameters, bypassing the fragile transfer function intermediate entirely. This enables our proposed learning algorithm to rigorously bound errors throughout the entire learning process, from raw measurements to final parameter estimates, with explicit polynomial dependence on the target precision $\varepsilon$. Crucially, we resolve the identifiability problem for generic Hamiltonians using smoothed analysis \cite{spielman2004smoothed}. For generic instances obtained by adding small random perturbations to any Hamiltonian, we establish identifiability up to physically unavoidable symmetries that cannot be resolved from any experiment (Theorem~\ref{thm:main_informal} and Theorem~\ref{thm:application_intro}). This transforms identifiability from an instance-specific property requiring case-by-case verification into a robust, generic feature of the problem family that holds for almost all Hamiltonians.

Methodologically, to our knowledge, \emph{quantitative algebraic-geometric} tools have not previously been employed to certify generic identifiability and fiber sizes in quantum learning. We combine these tools with anti-concentration bounds for polynomials and smoothed analysis \cite{carbery2001distributional,spielman2004smoothed}, yielding lower bounds that rule out near-collisions among solutions and certify well-conditioned Jacobians on generic instances; see also \cite{cox1998using} for background. A related but complementary approach, particularly prominent in the computer science literature, leverages \emph{sum-of-squares} (SoS) relaxations to solve polynomial systems arising in various classical statistical estimation tasks; we refer to the survey of \cite{raghavendra2018high} for an overview. Recently, this toolbox has been adapted to the setting of Gibbs state learning. In particular, learning from Gibbs states at any constant temperature can be formulated as a low-degree polynomial system for which low-level SoS provably recovers the parameters \cite{bakshi2024learning}. Our algebraic geometry approach is complementary to this SoS methodology. We further exploit explicit coefficient and degree bounds from high-temperature expansions \cite{haah2024learning} when converting probe measurements into low-degree polynomial constraints with controlled truncation error.

\section{Preliminaries}
\label{sec:prelims}

In this section we collect notation and various technical results that will be useful in subsequent sections. In Section~\ref{sec:smoothed}, we define the smoothed analysis model that we will work in. In Section~\ref{sec:series}, we introduce some tools for bounding terms in series expansions of time-evolved observables and of Gibbs states with respect to time and inverse temperature. In Section~\ref{sec:anticoncentration}, we record some standard results on anti-concentration of low-degree polynomials. In Section~\ref{sec:alggeo}, we present some tools from algebraic theory, more precisely from elimination theory, that will be crucial to Sections~\ref{sec:verify-generic-fiber} and~\ref{sec:quantitative}.

\paragraph{Basic notation.} Given $i\in[n]$ and $\mu\in\{0,1,2,3\}$, let $\sigma^\mu_i$ (where $\sigma_i^0 = I$, $\sigma_i^1 = X$, $\sigma_i^2 = Y$, and $\sigma_i^3 = Z$)  denote the $n$-qubit Pauli operator with index $\mu$ acting on the $i$-th qubit.  We will sometimes instead write $\mu \in \{I,X,Y,Z\}$. Given a Hamiltonian $H$, we use the following shorthand for time evolution under Hamiltonian $H$:
\begin{equation}
    \rho_H(t) \triangleq e^{-\i Ht} \rho \,e^{\i Ht}\,.
\end{equation}

\subsection{Smoothed Hamiltonians}
\label{sec:smoothed}

\begin{definition}\label{def:Ham_family}
    An \emph{$N$-parameter $\mathfrak{K}$-local Hamiltonian family} $(H_\params)$ associates to each $\params\in \R^N$ an $n$-qubit Hamiltonian $H_\theta = \sum_{P\in \{I,X,Y,Z\}^{\otimes n}} c_P(\params) P$, where $c_P: \R^N\to\R$ is a linear function\footnote{Linearity is not strictly speaking necessary, but it holds in all relevant examples in this work.} and $c_P \equiv 0$ unless $|P|\le \mathfrak{K}$. 
    
    Typically $\mathfrak{K}$ will be an absolute constant, in which case we sometimes simply say \emph{local} instead of $\mathfrak{K}$-local.
\end{definition}

\noindent For example if the $N$ indices of $\params$ correspond to all $\mathfrak{K}$-local Pauli operators, and $c_P(\params) = \lambda_P$ for all $|P| \le \mathfrak{K}$, this simply corresponds to the class of $\mathfrak{K}$-local Hamiltonians. In this work we will consider more structured Hamiltonian families that arise in physical contexts. Below we consider one such example:

\begin{example}[Translation-invariant nearest-neighbor Hamiltonians in 1D]
    Let $N = 12$, with $\params = (h_1, h_2, h_3, J_{11}, J_{12},\ldots,J_{33})$. The $12$-parameter, $2$-local Hamiltonian family of \emph{translation-invariant nearest-neighbor 1D Hamiltonians} associates to $\params$ the Hamiltonian
    \begin{equation}
        H_\params = \sum^n_{i=1} \sum^3_{\mu = 1}  h_i \sigma^\mu_i + \sum_{i = 1}^{n-1} \sum^3_{\mu,\nu = 1} J_{\mu\nu} \sigma^\mu_i\sigma^\nu_{i+1}\,.
    \end{equation}
\end{example}

\noindent We are now ready to define our probabilistic model for how the ground truth Hamiltonian is generated.

\begin{definition}[Smoothed analysis]\label{def:smoothed_analysis}
    Let $(H_\params)$ be some $N$-parameter local Hamiltonian family, and let $\mu\in\R^N$ be unknown. Given \emph{smoothing parameter} $\smoothing$, a \emph{$\smoothing$-smoothed Hamiltonian} is the random Hamiltonian $H_{\params^*}$ where $\params^* \sim \mathcal{N}(\mu,\smoothing^2)$.
\end{definition}

\noindent It will be convenient to define an ``effective radius'' for the space of possible parameters $\params^*$. By standard Gaussian concentration, $\norm{\params^*}_\infty \le \norm{\mu}_\infty + 2\sigma\sqrt{\log N/\delta}$ with probability at least $1 - \delta$. We will consider smoothed Hamiltonians for which $\norm{\mu}_\infty \le \mumax$, so henceforth define
\begin{equation}
    \lmax \triangleq \mumax + 2\sigma\sqrt{\log N/\delta}\,.
\end{equation}

\subsection{Series Expansions in Time and Inverse Temperature}
\label{sec:series}

As is standard in this literature, we use series expansions in time and inverse temperature to make sense of time-evolved observables and Gibbs states.

\begin{lemma}[Hadamard formula]\label{lem:hadamard}
    For any $H,X$,
    \begin{equation}
        \Bigl\|e^{\i Ht}Xe^{-\i Ht} - \sum^{s-1}_{j=0} \frac{1}{j!} (\i t)^j [H,X]_j\Bigr\|_F \le \frac{t^s}{s!}\norm{[H,X]_s}_F\,,
    \end{equation}
    where $[H,X]_j = [H,[H,X]_{j-1}]$ denotes the $j$-fold nested commutator.
\end{lemma}

\begin{theorem}[Theorem 3.1 from~\cite{haah2024learning}]\label{thm:hightemp}
    Given an $n$-qubit local Hamiltonian $H = \sum^M_{a=1} \lambda_a \sigma_a$ with Gibbs state $\rho_\beta\propto \exp(-\beta H)$ in the $d=2^n$-dimensional Hilbert space, we have the Taylor expansion
    \begin{equation}
        \tr(\sigma_a \rho_\beta) = \frac{1}{d}\,\tr(\sigma_a) + \sum^\infty_{k=1} \beta^k p_k(\lambda_1,\ldots,\lambda_M)\,,
    \end{equation}
    %\sitan{SC TODO: change $m$ to $k$ to be consistent with last section notation, and then percolate the $m$'s to the boilerplate section}
    with equality whenever the series converges absolutely. Suppose the dual interaction graph for $H$ has degree $\mathfrak{d}$. For any $k\ge 1$, $p_k$ is a degree-$k$ homogeneous polynomial in $\lambda_1,\ldots,\lambda_M$ consisting of at most $e\mathfrak{d}(1 + e(\mathfrak{d} - 1))^k$ monomials, where the coefficient in front of any monomial is at most $(2e(\mathfrak{d} + 1))^{k+1}(k+1)$ in magnitude.
    Furthermore, the sum of the absolute values of the coefficients of $p_k$ is bounded by 
    \begin{equation}
        c_k \triangleq 2e^2\mathfrak{d}(\mathfrak{d}+1)\tau^k(k+1)
    \end{equation}
    for $\tau = (1 + e(\mathfrak{d} - 1))(2e(\mathfrak{d}+1)) \le 2e^2(\mathfrak{d}+1)^2\,,$
    and thus for any $-\lmax\le \lambda_1,\ldots,\lambda_M \le \lmax$ we have
    \begin{equation}
        |\partial_{\lambda_a} p_k| \le \exp(\Od(k))
    \end{equation}
    for all $a\in[M]$.
    Suppose further that $\sigma_a$ is an $L$-local Pauli operator. Then after $\mathcal{O}(LM\mathfrak{d}\log\mathfrak{d})$ pre-processing time, for all $k$ we have that the terms of the polynomial $p_k$ can be computed in time $\mathcal{O}((k\mathfrak{d} + (8^k + L)\poly(k))\cdot \mathfrak{d}(1+e(\mathfrak{d}-1)^k))$.
\end{theorem}

\noindent When $\mathfrak{d} = \mathcal{O}(1)$ (e.g.~in the 1D case, $\mathfrak{d} = 4$), $c_k = \exp(Ck)$ for some absolute constant $C > 0$, and the runtime to compute all of the terms in the Taylor expansion up to degree $k$ is $\mathcal{O}(L\exp(C'k))$ for some absolute constant $C'>0$.

\subsection{Anti-Concentration Tools}
\label{sec:anticoncentration}

A fundamental tool in smoothed analysis is the fact that low-degree polynomials in Gaussian random variables anti-concentrate.

\begin{lemma}[Carbery-Wright]\label{lem:carbery}
    Let $p: \R^n\to\R$ be a degree-$d$ polynomial. Then for any $u > 0$,
    \begin{equation}
        \Pr[g\sim\calN(0,\Id)]{|p(g) - u| \le \epsilon\cdot \Var(p)^{1/2}} \le \mathcal{O}(\epsilon^{1/d})\,,
    \end{equation}
    where $\Var(p)$ denotes the variance of $p(g)$ for $g\sim\calN(0,\Id)$.
\end{lemma}

\noindent We will use the following basic lower bound for the variance of a polynomial under a smoothed input, in terms of the leading terms in its Hermite expansion.

\begin{lemma}\label{lem:anticonc}
    Let $f: \R^p\to\R$ be a non-constant, degree-$t$ polynomial with normalized Hermite expansion $f = \sum_\alpha b_\alpha h_\alpha$, and let $\sum_{\alpha \ \mathrm{maximal}} b_\alpha^2$ denote the sum of squares of the coefficients for $\alpha$ which are maximal among terms appearing in the Hermite expansion. Then for $h\sim \mathcal{N}(\mu,\sigma)$, we have
    \begin{equation}
        \Var{f(h)} \ge \sigma^{2t} \sum_{\alpha \ \mathrm{maximal}} b_\alpha^2\,.
    \end{equation}
\end{lemma}

\begin{proof}
    This follows immediately from the fact that $h_\alpha(\mu + \sigma g) = \sigma^{|\alpha|} h_\alpha(g) + \sum_{\beta \prec \alpha} e_\beta h_\beta(g)$.
\end{proof}

\subsection{Algebraic Geometry Tools}
\label{sec:alggeo}

In this section, we collect a key tool from algebraic geometry, specifically elimination theory, that will be crucial for showing quantitative identifiability in the polynomial systems that we design.

Throughout this section, let $S = \C[t_1,\ldots,t_m]$ denote a polynomial ring generated by ``coefficient variables'' $t_1,\ldots,t_m$, and let $q_1,\ldots,q_N \in S[X_1,\ldots,X_N]$. In our application, we will consider $q_1,\ldots,q_N$ for which only the constant terms are non-scalar elements of $S$, but in this section we work in full generality.

We will often consider \emph{specializations} $\phi: S\to \C$ which associate to every coefficient variable $t_1,\ldots,t_m$, and thus every polynomial thereof, a scalar quantity. Under such a $\phi$, given polynomials $q_1,\ldots,q_N \in S[X_1,\ldots,X_N]$, we denote by $\phi(q_i)$ their specializations to polynomials in $\C[X_1,\ldots,X_N]$. We will be interested in the size of the common zero set $\{\phi(q_1) = \cdots = \phi(q_N) = 0\}$ for ``generic'' choices of $\phi$.

\newcommand{\fibersize}{D}
\begin{definition}[Generic finiteness]\label{def:finiteness}
    We say that the polynomial system $\{q_1 = \cdots = q_N = 0\}$ satisfies \emph{generic finiteness} if there is an integer $\fibersize$ and a dense Zariski-open $U\subseteq \C^m$ such that for every $u\in U$, if $\phi(t_i) = u_i$, then the \emph{affine fiber} $V_\phi$ - i.e. the common zero set $\{\phi(q_1) = \cdots = \phi(q_N) = 0\}$ - is of size $\fibersize$.
\end{definition}

\begin{lemma}[Product formula]\label{lem:poisson}
    Let $f, q_1,\ldots,q_N\in S[X_1,\ldots,X_N]$, and suppose the $q_i$'s satisfy generic finiteness (Definition~\ref{def:finiteness}). Then the function
    \begin{equation}
        \prod_{\bx\in V_\phi} (t - \phi(f)(\bx))^{m(\bx)}\,, \label{eq:function_of_coeffs}
    \end{equation}
    regarded as a function in $t$ and the coefficient variables $t_1,\ldots,t_m$, is an element of $S[w]$, that is, a polynomial in $t, t_1,\ldots,t_m$. Here $m(\bx)$ denotes the scheme-theoretic multiplicity of $\bx$ on $V_\phi$, which is $1$ if the Jacobian of the map $(\phi(q_1),\ldots,\phi(q_N))$ is invertible at $\bx$.
\end{lemma}

\noindent We defer a proof of this to Appendix~\ref{app:productproof}. We remark that this is a generalization of the \emph{Poisson product formula}~\cite[Chapter 3.3, Theorem 3.4]{cox1998using}, which concerns the special case where $q_1,\ldots,q_N$ share no projective zeroes at infinity, in which case the function in Eq.~\eqref{eq:function_of_coeffs} is, up to scaling, the \emph{resultant} $\mathrm{Res}(f^{\sf hom},p^{\sf hom}_1\ldots,p^{\sf hom}_N)$ of the homogenizations of the polynomials. In our applications, $q_1,\ldots,q_N$ do actually share projective zeroes at infinity however, necessitating the use of our bespoke Lemma~\ref{lem:poisson}.

Finally and separately, we will need the following quantitative version of the classic inverse function theorem. This is an analytic statement rather than an algebraic one, but we will apply it exclusively to polynomial maps.

\begin{lemma}[Quantitative inverse function theorem]\label{lem:inverse}
    Let $F: \C^N\to\C^N$ be a continuous map, and for $\bx\in \C^N$, define $\by = F(\bx)$. Suppose $J_F(\bx)$ is nonsingular.
    
    Let $\eta,\smoothness > 0$ be parameters, with $\eta < \sigma_{\min}(J_F(\bx))/\smoothness$ and define 
    \begin{equation}
        \xi \triangleq \sigma_{\min}(J_F(\bx)) \eta - \frac{1}{2}\eta^2\smoothness > \frac{1}{2}\sigma_{\min}(J_F(\bx))\eta\,. \label{eq:xieta}
    \end{equation} 
    Suppose $\norm{J_F(\bx') - J_F(\bx'')}_{\sf op} \le \smoothness \norm{\bx' - \bx''}_2$ for all $\bx',\bx''\in B_\eta(\bx)$
    Then for any $\tilde{\by} \in B_\xi(\by)$, there is a unique point $\tilde{\bx} \in B_\eta(\bx)$ for which $F(\tilde{\bx}) = \tilde{\by}$.
\end{lemma}

\begin{proof}
    Consider the map $H$ that sends $\vdelta$ to $J_F(\bx)^{-1}(\tilde{\by} - \by - \mathrm{rem}(\vdelta))$, where $\mathrm{rem}$ is the Taylor remainder
    \begin{equation}
        \mathrm{rem}(\vdelta) \triangleq F(\bx + \vdelta) - F(\bx) - J_F(\bx)\delta = \int^1_0 (J_F(\bx + t\vdelta) - J_F(\bx)) \vdelta\, \mathrm{d}t\,.
    \end{equation}
    If $\norm{\vdelta}_2 \le \eta$, then $\norm{\mathrm{rem}(\vdelta)}_2 \le \int^1_0 \eta^2\smoothness t\,\mathrm{d}t = \frac{1}{2}\eta^2\smoothness$. By assumption $\norm{\tilde{\by} - \by}_2 \le \xi$, so 
    \begin{equation}
        \norm{H(\vdelta)}_2 \le \sigma_{\min}(J_F(\bx))^{-1} (\xi + \frac{1}{2}\eta^2\smoothness) = \eta
    \end{equation}
    by Eq.~\eqref{eq:xieta}. In other words, $H(B_\eta(0)) \subseteq B_\eta(0)$. Furthermore, one can show that it is a contraction. Indeed, for any $\vdelta,\vdelta' \in B_\eta(0)$,
    \begin{align}
        \MoveEqLeft \norm{H(\vdelta) - H(\vdelta')}_2 \\
        &= \norm{J_F(\bx)^{-1}(\mathrm{rem}(\vdelta') - \mathrm{rem}(\vdelta))}_2 \\ 
        &\le \sigma_{\min}(J_F(\bx))^{-1} \norm{\mathrm{rem}(\vdelta') - \mathrm{rem}(\vdelta)}_2 \\
        &\le \sigma_{\min}(J_F(\bx))^{-1} \int^1_0 \Bigl(\norm{(J_F(\bx + t\vdelta') - J_F(\bx))(\vdelta'-\vdelta)}_2 + \norm{(J_F(\bx + t\vdelta') - J_F(\bx + t\vdelta))\vdelta}_2\Bigr)\,\mathrm{d}t\\
        &= \sigma_{\min}(J_F(\bx))^{-1} \smoothness\eta\norm{\vdelta' - \vdelta}_2\,.
    \end{align}
    By our assumption that $\eta < \sigma_{\min}(J_F(\bx))/\smoothness$, we conclude that $H$ is a contraction on $B_\eta(0)$. By Banach's fixed-point theorem, it has a unique fixed point $\vdelta^\circ \in B_\eta(0)$. Multiplying both sides of $H(\vdelta^\circ) = \vdelta^\circ$ by $J_F(\bx)^{-1}$ and rearranging, we find that
    \begin{equation}
        \tilde{\by} - \by = F(\bx + \vdelta^\circ) - F(\bx)\,.
    \end{equation} As $\by = F(\bx)$, if we define $\tilde{\bx} = \bx + \vdelta^\circ$ then the claim follows.
\end{proof}

\subsection{Proof of Lemma~\ref{lem:poisson}}
\label{app:productproof}

Let $S = \C[t_1,\ldots,t_m]$. and let $q_1\ldots,q_N \in S[X_1,\ldots,X_N]$ as in Section~\ref{sec:alggeo}. Introducing a homogenizing variable $X_0$, denote their homogenizations by $q^{\sf hom}_1,\ldots,q^{\sf hom}_N \in S[X_0,\ldots,X_N]$.

Let $I\subset S[X_1,\ldots,X_N]$ denote the ideal generated by $q_1,\ldots,q_N$, and let $I^h\subseteq S[X_1,\ldots,X_N]$ denote the ideal generated by $q^{\sf hom}_1,\ldots,q^{\sf hom}_N$.

\begin{definition}[Saturated ideals]
    Given ideals $I, J$ in $S[X_1,\ldots,X_N]$, the \emph{ideal quotient} $(I:J)$ is the ideal consisting of $f\in S[X_1,\ldots,X_N]$ for which  $fJ \subseteq I$. The \emph{saturation} of $I$ with respect to $J$ is $(I: J^\infty) = \cup_{k\ge 1} (I: J^k)$. As $S$ is \emph{Noetherian}, there is some finite $\overline{k}$ such that $(I: J^\infty) = (I: J^{\overline{k}})$.
\end{definition}

\noindent The polynomials $q_1,\ldots,q_N$ may have solutions ``at infinity,'' which we will remove using saturation. First, compute $I' \triangleq (I : X_0^\infty) = (I : X_0^{\overline{k}})$, and then de-homogenize by specializing the $X_0$ variable to $1$, thereby obtaining the new ideal $I^* \subseteq S[X_1,\ldots,X_N]$. Define the coordinate ring
\begin{equation}
    A \triangleq S[X_1,\ldots,X_N] / I^*\,.
\end{equation}
In this notation, generic finiteness is equivalent to the condition that under any specialization as in Definition~\ref{def:finiteness}, we have that $\dim(A\otimes_S \C) = \fibersize$.

\begin{definition}[Localization]
    Given a ring $A$ and an element $g\in A$, the \emph{localization} of $A$ by $g$, denoted $A_g$, is the ring consisting of elements of the form $a/g^i$, where $a\in A$ and $i\in\mathbb{Z}_{\ge 0}$, with addition and multiplication defined in the natural way, and with $a/g^i$ and $a'/g^{i'}$ equivalent if there exists $j$ such that $g^j(a g^{i'} - a' g^i) = 0$.
\end{definition}

\begin{lemma}[Generic freeness]
    If $q_1,\ldots,q_N$ satisfy generic finiteness (Definition~\ref{def:finiteness}), there exists nonzero $g\in S$ for which the localization $A_g$ is a free $S_g$-module of rank $\fibersize$.
\end{lemma}

\begin{proof}
    Note that $A$ is a finitely generated $S$-algebra as it is a quotient of $S[X_1,\ldots,X_N]$, so by Grothendieck's generic freeness lemma~\cite[Lemma 6.9.2]{grothendieck1965elements}, there exists nonzero $g\in S$ for which $A_g$ is a free $S_g$-module. It remains to show that the rank is $\fibersize$. Take a generic specialization $S\to \C$ and note that by generic finiteness, $\dim(A\otimes_S \C) = \fibersize$. Additionally, $A\otimes_S \C\cong A_g \otimes_{S_g} \C$, so we conclude that $A_g \cong S_g^{\otimes \fibersize}$ as claimed.
\end{proof}

\noindent Given $f\in S[X_1,\ldots,X_N]$, denote by $m_f: A_g\to A_g$ the \emph{multiplication map} which sends $h\in A$ to $fh$. 

\begin{lemma}\label{lem:charpoly}
    The \emph{characteristic polynomial} \begin{equation}
    \chi_g(t) \triangleq \det(t\Id - m_f)
    \end{equation} satisfies the following properties:
    \begin{itemize}
        \item $\chi_g\in S_g[t]$
        \item For any $\phi: S_g\to \C$ for which the composition $S\to S_g\to \C$ yields a finite affine fiber $V_\phi$ of size $\fibersize$, if $\phi(\chi_g) \in \C[t]$ is the result of specializing the coefficients of $\chi_g$ via $\phi$, then $\phi(\chi_g)(t) = \prod_{\bx\in V_\phi} (t - f(\bx))^{m_{\bx}}$, where $m_{\bx}$ is the scheme-theoretic multiplicity of $\bx$ in $V_\phi$.
    \end{itemize}
\end{lemma}

\begin{proof}
    $A_g$ is a free $S_g$-module of rank $\fibersize$, so $\chi_g$ is well-defined and the first step is immediate. The second step follows by the standard fact \cite[Chapter 2.4, Theorem 4.5]{cox1998using} that the eigenvalues of the multiplication map correspond to $f(\bx)$ for the different $\bx\in V_\phi$, with multiplicity.
\end{proof}

\noindent Because $\phi(\chi_g)(t) = \prod_{\bx\in V_\phi} (t - f(\bx))^{m_{\bx}}$ for generic specializations $S_g \to \C$ by Lemma~\ref{lem:charpoly}, we conclude that $\prod_{\bx\in V_\phi} (t - f(\bx))^{m_{\bx}}$ is equal to $\chi_g(t)$ as an element of $S_g[t]$. It remains to conclude that $\chi_g(t)$ is an element of $S[t]$, not just $S_g[t]$. For this, we appeal to the determinantal trick.

\begin{lemma}[Determinantal trick]\label{lem:det}
    For any finitely generated $S$-module $A$ and ideal $I\subseteq S$, if $\phi: A\to A$ is an $S$-module endomorphism satisfying $\phi(A) \subseteq I\cdot A$, then there is a monic polynomial $P\in S[t]$ for which $P(\phi) = 0$.
\end{lemma}

\begin{proof}
    This is standard, see e.g., \cite[Proposition 2.4]{atiyah2018introduction}.
\end{proof}

\noindent From this we can deduce the following, which concludes the proof of Lemma~\ref{lem:poisson}:

\begin{corollary}
    $\chi_g(t)\in S[t]$.
\end{corollary}

\begin{proof}
    Note that $m_f$ satisfies the properties needed in the endomorphism $\phi$ in Lemma~\ref{lem:det} with $I = S$, so there exists monic polynomial $P\in S[t]$ for which $P(m_f) = 0$. But this implies that the eigenvalues of $m_f$ are roots of $P$ and thus integral. As $S$ is integrally closed, all integral elements are elements of $S$. The coefficients of $\chi_g(t)$ are symmetric polynomials in these eigenvalues and thus also elements of $S$, so $\chi_g\in S[t]$ as desired.
\end{proof}

\section{Local Convergence Guarantee}
\label{sec:local}

Let $f: \R^N\to \R^m$ denote a differentiable function. In our applications, on input $\bx\in\R^N$ which parametrizes some Hamiltonian $H$, each $f_i(\bx)$ computes the difference between an estimate of some observable value obtained from measurements, and the true value of a certain ``truncation'' of the observable. Because of measurement error and truncation error, for $\bx^*$ the true parameters of the Hamiltonian generating the data, $f(\bx^*)$ is small but nonzero. Generally the goal is to (1) find a choice of parameters $\bx$ for which $f(\bx)$ is sufficiently small, and (2) show that under some additional conditions (e.g., a loose upper bound on $\norm{\bx - \bx^*}_\infty$), this implies $\bx$ is close to $\bx^*$.

We will denote the output coordinates of $f$ by $f_1,\ldots,f_m$, the Jacobian of $f$ at $\bx$ by $J(\bx)\in \R^{m\times N}$, and the Hessian of $f_i$ at $\bx$ by $\Hess_i(\bx) \in \R^{N\times N}$. We will also use the shorthand $(J^+f)(\bx) = J(\bx)^+\cdot f(\bx)\in \R^{N}$, where $A^+ \triangleq (A^\top A)^{-1} A^\top$ denotes the Moore-Penrose pseudo-inverse of $A$.

The (projected) Newton's method starts at a point $\bx^{(0)} \in\mathbb{R}^N$ and applies the following iterative update
\begin{equation}
    \bx^{(t+1)} = \proj_{[-\lmax,\lmax]^N} \Bigl(\bx^{(t)} - (J^+f)(\bx^{(t)})\Bigr)\,.
\end{equation}

In this section we give a standard analysis of Newton's method. This was already used in the work of~\cite{haah2024learning} to obtain an optimal algorithm for learning high-temperature Gibbs states. We include a version of this analysis below for completeness.

\begin{theorem}\label{thm:newton}
    Suppose $f: \R^N\to\R^m$ satisfies the following conditions:
    \begin{enumerate}
        \item \underline{Near vanishing at ground truth:} $\norm{f(\bx^*)}_\infty \le \epsilon$
        \item \underline{Smoothness everywhere:} $\max_{j\in [m], \vmu\in[-\lmax,\lmax]^N} |\vDelta^\top \Hess_j(\vmu) \vDelta| \le \smoothness \norm{\vDelta}^2_\infty$ for all $\vDelta\in\R^N$
        \item \underline{Curvature near ground truth:} $\sigma_{\min}(J(\bx^*)) \ge \curvature$
    \end{enumerate}
    If
    \begin{equation}
        \epsilon \le \curvature^2 / 8m \label{eq:eps}
    \end{equation}
    and
    \begin{equation}
        \norm{\bx^{(0)} - \bx^*}_\infty \le \frac{\curvature}{2\smoothness\sqrt{m N}} \label{eq:init}\,,
    \end{equation}
    then for any $t \ge \log_2(\curvature^2/m\epsilon) - 1$, we have 
    \begin{equation}
        \norm{\bx^{(t)} - \bx^*}_\infty \le 6\epsilon\curvature^{-1}\sqrt{m}\,.
    \end{equation}
\end{theorem}

\begin{proof}
    We need to track how the error contracts in one iteration. Define $\vDelta^{(t)} \triangleq \bx^{(t)} - \bx^*$. For most of the proof, we will abbreviate $\bx = \bx^{(t)}$, $\bx' = \bx^{(t+1)}$, $\vDelta = \bx - \bx^*$, and $\vDelta' = \bx' - \bx^*$. Then for any $i\in[N]$,
    \begin{align*}
        |\Delta'_i| &= |\proj_{[-\lmax,\lmax]}(x_i - (J^+f)(\bx)_i) - x^*_i| \\
        &\le |x_i - (J^+f)(\bx)_i - x^*_i|\,.
    \intertext{
        By Taylor remainder theorem, for every $j\in[m]$ there is $\vmu^{[j]}$ on the line segment between $\bx$ and $\bx^*$ such that 
        \begin{equation}
            f_j(\bx^*) = f_j(\bx) - \sum^N_{i=1} \Delta_i J(\bx)_{ji} + \frac{1}{2} \vDelta^\top \Hess_j(\vmu^{[j]}) \vDelta\,.
        \end{equation}
        Substituting this above yields
    }
        &= \Bigl|\Delta_i - \sum^m_{j=1} J^+(\bx)_{ij} \Bigl(f_j(\bx^*) + \sum^N_{i'=1}\Delta_{i'} J(\bx)_{ji'} - \frac{1}{2}\vDelta^\top \Hess_j(\vmu^{[j]})\vDelta\Bigr)\Bigr| \\
        &= \Bigl|\Bigl(\vDelta - (J^+(\bx) f(\bx^*)) + (J^+ J)(\bx) \vDelta\Bigr)_i + \frac{1}{2} \sum^m_{j=1} J^+(\bx)_{ij} \vDelta^\top \Hess_j(\vmu^{[j]}) \vDelta\Bigr| \\
        &\le \norm{J^+(\bx)}_{\infty\to\infty} \norm{f(\bx^*)}_\infty + \frac{1}{2}\norm{J^+(\bx)}_{\infty\to\infty}\cdot \max_{j\in[m], \vmu\in[-\lmax,\lmax]^N} |\vDelta^\top \Hess_j(\vmu) \vDelta| \\
        &\le \sqrt{m}\sigma_{\min}(J(\bx))^{-1}\cdot (\epsilon + \frac{1}{2}\smoothness\norm{\vDelta}^2_\infty)\,.
    \end{align*}
    By Weyl's inequality, $\sigma_{\min}(J(\bx)) \ge \curvature - \norm{J(\bx) - J(\bx^*)}_{\sf op}$. It is straightforward to obtain a bound on the Lipschitz constant for $J$: for each $j\in[m]$, we have that 
    \begin{align}
        \norm{J(\bx) - J(\bx^*)}_{\sf op} &\le \sqrt{m}\max_{j\in[m]} \norm{\nabla f_j(\bx) - \nabla f_j(\bx^*)}_2 \\
        &\le \sqrt{m}\max_{j\in[m], \vmu\in[-\lmax,\lmax]^N}\norm{\Hess_j(\vmu)}_{\sf op} \cdot \norm{\vDelta}_2 \\
        &\le \smoothness\sqrt{mN} \norm{\vDelta}_\infty\,.
    \end{align}

    As long as we can ensure that $\norm{\vDelta}_\infty \le \frac{\curvature}{2\upsilon\sqrt{mN}}$ across all iterations, then we have the recursion
    \begin{equation}
        \norm{\vDelta'} \le \alpha + \beta\norm{\vDelta}^2_\infty\, \qquad \text{for} \qquad \alpha = 2\epsilon\curvature^{-1}\sqrt{m}\ \ \text{and} \ \ \beta = \upsilon\curvature^{-1}\sqrt{m}\,.
    \end{equation}
    for each iteration. 

    Applying Fact~\ref{fact:recursion} below with this choice of $\alpha,\beta$ (and verifying that $\alpha\beta \le 1/4$ by Eq.~\eqref{eq:eps}) and $\overline{z} \triangleq \frac{\curvature}{2\upsilon\sqrt{mN}}$, 
    we conclude that by our assumed bound on $\norm{\vDelta^{(0)}}_\infty$ in Eq.~\eqref{eq:init}, we indeed have $\norm{\vDelta^{(t)}}_\infty \le \overline{z}$ for all $t$ and furthermore $\norm{\vDelta^{(t)}}_\infty \le 3\alpha = 6\epsilon\curvature^{-1}\sqrt{m}$ provided $t \ge \log_2 \frac{1}{\alpha\beta} - 1 = \log_2(\curvature^2/m\epsilon) - 2$.
\end{proof}

\noindent The above proof makes use of the following elementary calculation:

\begin{fact}\label{fact:recursion}
    Let $\alpha,\beta \in \R_{> 0}$ be such that $\alpha\beta \le 1/4$. Let $2\alpha \le \overline{z} \le 1/2\beta$. Given a sequence of positive reals $z_0,z_1,z_2,\ldots$ satisfying $z_0 \le \overline{z}$ and $z_{t+1} \le \alpha + \beta z^2_t$ for all $t\ge 0$, then 
    \begin{equation}
        z_t \le \min(\overline{z}, 2\alpha + \frac{1}{\beta 2^{t + 1}})\,.
    \end{equation}
    In particular, for $t \ge \log_2 \frac{1}{\alpha\beta} - 1$, we have that $z_t \le 3\alpha$.
\end{fact}

\begin{proof}
Note that $\sqrt{\overline{z}/2\beta} \le \overline{z}$ as $\overline{z} \le 1/2\beta$. We thus have that 
\begin{equation}
z_1 \le \alpha + \beta z^2_0 \le \overline{z}\,,
\end{equation}
as $2\alpha \le \overline{z}$ and $z_0 \le \overline{z} \le \sqrt{\overline{z}/2\beta}$. Iterating this, we conclude that $z_t \le \overline{z}$ for all $t$. The bound of $z_t \le 2\alpha + \frac{1}{\beta 2^{t+1}}$ is shown in~\cite[Lemma 4.7]{haah2024learning}.
\end{proof}

\section{Estimating Taylor Series Coefficients}
\label{sec:estimation}

Having established necessary analytic and algebraic preliminaries, we now turn
to concrete quantities that arise in our quantum probe tomography setting. In particular, our analysis will require a careful understanding of how expectation values of local observables expand in both the evolution time $t$ and the inverse temperature $\beta$. To this end, we introduce a systematic way to define and estimate the Taylor series coefficients of these observables, and we analyze both their algebraic structure and estimability from experiments.

In this work we will use observables of the form
\begin{equation}
\label{E:observablesA0}
    \obsvalue{\beta}{t}{\mu}{C} \triangleq \tr(\sigma^\mu_{\probeindex} (C[\rho_\beta])_H(t))\,,
\end{equation}
where $\sigma^\mu_{\probeindex}$ is a single-qubit Pauli operator acting on qubit $\probeindex$ labeled by $\mu \in \{X, Y, Z\}$, $\rho_{\beta} \propto e^{-\beta H}$ is the Gibbs state of the underlying Hamiltonian $H$, and $C$ is a single-qubit control channel applied to $\rho_\beta$.  Given $j,k\in\mathbb{Z}_{\ge 0}$, we let $\obsvalue{\beta^{(k)}}{t^{(j)}}{\mu}{C}$ denote the coefficient of $t^j\beta^k$ in the Taylor series expansion of $\obsvalue{\beta}{t}{\mu}{C}$.

The following shows that when one regards these Taylor coefficients formally as polynomials in the parameters $\params$ of the Hamiltonian $H$, these polynomials have bounded entries and can be written down in time scaling exponentially in the degree $j+k$.

\begin{lemma}\label{lem:fewpolys}
     For every $j,k\in\mathbb{Z}_{\ge 0}$, the Taylor coefficient $\obsvalue{\beta^{(k)}}{t^{(j)}}{\mu}{C}$ is a polynomial in the parameters $\params = (\lambda_1,\ldots,\lambda_N)$ of the underlying local Hamiltonian, and the coefficients of that polynomial have magnitudes summing to at most $\mathcal{O}(\exp(\Od(j+k)))$ and can all be computed in time $\mathcal{O}(\exp(\Od(j+k)))$, where $\mathfrak{d}$ is the degree of the dual interaction graph of $H$.\footnote{The \emph{dual interaction graph} of a local Hamiltonian $H$ is a graph whose vertices correspond to terms in the Hamiltonian, and edges connect terms whose supports overlap on at least one qubit. We say that $H$ is \emph{low-intersection}if the degree $\mathfrak{d}$ of this graph is $\mathcal{O}(1)$.}
\end{lemma}

\begin{proof}
    We can first Taylor expand $\obsvalue{\beta}{t}{\mu}{C} = \tr(C^\dagger[e^{\i Ht} \sigma^\mu_\probeindex e^{-\i Ht}] \cdot \rho_\beta)$ in $t$ to get 
    \begin{equation}
        \obsvalue{\beta}{t}{\mu}{C} = \sum^\infty_{j=0} \frac{1}{j!} (\i t)^j \, \tr(C^\dagger[[H, \sigma^\mu_\probeindex]_j] \cdot \rho_\beta) \label{eq:series_t}
    \end{equation}
    by the Hadamard formula (Lemma~\ref{lem:hadamard}). Note that the Pauli expansion of $[H, \sigma^\mu_\probeindex]_j$ consists of terms which are at most $j$ hops from site $\probeindex$, of which there are at most $\exp(\Od(j))$ each of locality $\mathcal{O}(j)$. So the Taylor coefficient $\obsvalue{\beta^{(k)}}{t^{(j)}}{\mu}{C}$ is a linear combination with $\exp(\mathcal{O}(j))$-bounded coefficients of at most $\exp(\Od(j))$ many polynomials of the form $p_k$ from Theorem~\ref{thm:hightemp}, each of which can be computed in time $\mathcal{O}(j\exp(\Od(k)))$ and has coefficients summing to $\exp(\Od(k))$.
\end{proof}

\noindent Next, we show that for the particular Hamiltonian to which we have probe access, we can estimate any Taylor coefficient $\obsvalue{\beta^{(k)}}{t^{(j)}}{\mu}{C}$ from experiments.

We begin by showing that we can estimate the time derivatives of $\obsvalue{\beta}{t}{\mu}{C}$ through finite differencing. Define
\begin{equation}
    \newdelt{j}{\beta}{\mu}{C}\triangleq \frac{\i^j}{j!}\tr(C^\dagger[[H,\sigma^\mu_\probeindex]_j]\cdot \rho_\beta)\,. \label{eq:newdelt}
\end{equation}

\begin{lemma}\label{lem:deltestimate}
    For any $\beta > 0$, given probe access to $H$ at inverse temperature $\beta$, one can estimate $\newdelt{j}{\beta}{\mu}{C}$ for any $\mu\in\{X,Y,Z\}$ and any single-qubit control channel $C$ to additive error $\mathcal{O}(t\exp(\Od(j)))$ with probability at least $1 - \delta$ using $n_{\sf query} = \Theta_{\mathfrak{d}}(t)^{-2j-2}\log 1/\delta$ queries and total evolution time $j^2 \Theta_{\mathfrak{d}}(t)^{-2j-1}\log1/\delta$.
\end{lemma}

\begin{proof}
    Recalling the expansion of $\obsvalue{\beta}{t}{\mu}{C}$ in $t$ from Eq.~\eqref{eq:series_t}, by Lemma~\ref{lem:hadamard} we have
    \begin{equation}
        \Bigl|\frac{\i^j}{j!}\tr(C^\dagger[[H,\sigma^\mu_\probeindex]_j]\cdot \rho_\beta) - \frac{\i^j}{t^j} \sum^j_{\ell = 0} (-1)^{j-\ell}\binom{j}{\ell} \obsvalue{\beta}{\ell t}{\mu}{C}\Bigr| \le \mathcal{O}(t\exp(j+1) \norm{[H,\sigma^\mu_\probeindex]_{j+1}}_{\sf op})\,.
    \end{equation}
    Note that $\norm{[H,\sigma^\mu_\probeindex]_{j+1}}_{\sf op} \le \exp(\Od(j))$.
    One can estimate $\obsvalue{\beta}{t}{\mu}{C}$ for any $t$ to additive error $\epsilon$ with probability at least $1 - \delta$ using $\mathcal{O}(\log(1/\delta)/\epsilon^2)$ queries and total evolution time $t$, so by taking $\epsilon = \Od(t)^{j+1}$, we can estimate $\newdelt{j}{\beta}{\mu}{C}$ to additive error $\mathcal{O}(t\exp(\Od(j)))$ using $n_{\sf query} = \Theta_{\mathfrak{d}}(t)^{-2j-2}\log 1/\delta$ queries and total evolution time $\mathcal{O}(j^2 n_{\sf query} t) = j^2 \Theta_{\mathfrak{d}}(t)^{-2j-1} \log 1/\delta$ as claimed.
\end{proof}

\noindent The above Lemma shows how to estimate the normalized time derivatives $\newdelt{j}{\beta}{\mu}{C}$ at fixed $\beta$. To recover the full Taylor coefficients of $\obsvalue{\beta^{(k)}}{t^{(j)}}{\mu}{C}$ in both $t$ and $\beta$, we must additionally expand in the inverse temperature and control the associated truncation errors. The next result formalizes this step.

%\noindent Next, we show how to use finite differencing on $\newdelt{j}{\beta}{\mu}{C}$ to estimate the Taylor coefficients $\obsvalue{\beta^{(k)}}{t^{(j)}}{\mu}{C}$.

\begin{lemma}\label{lem:derivestimate}
    There is a constant $\beta_{\sf crit}$ depending only on the locality of $H$ and the degree of its dual interaction graph such that the following holds. Let $j,k\in\mathbb{Z}_{\ge 0}$, and let $\beta \le \beta_{\sf crit} / k$. Given probe access to $H$ at temperatures no higher than $1/\beta$, one can estimate $\obsvalue{\beta^{(k)}}{t^{(j)}}{\mu}{C}$ for any $\mu\in\{X,Y,Z\}$ and any single-qubit control channel $C$ to additive error $\nu + \beta\, \Od(k)^{k+1} \exp(\Od(j))$ with probability at least $1 - \delta$ using $\exp(\Od(j(j+k)))\cdot \mathcal{O}((\nu\beta^k)^{-2j-2} \log 1/\delta)$ queries with total evolution time $\exp(\Od(j(j+k)))\cdot \mathcal{O}((\nu\beta)^{-2j-1}  \log 1/\delta)$.
\end{lemma}

\noindent In particular, when $j,k,\mathfrak{d} = \mathcal{O}(1)$, the corresponding Taylor coefficient can be estimated to error $\nu + \mathcal{O}(\beta)$ in $\poly(1/\beta,1/\nu)$ queries and evolution time.

\begin{proof}
     For $\beta > 0$, let the estimates for $\newdelt{j}{\beta}{\mu}{C}$ from Lemma~\ref{lem:deltestimate} be denoted by $\hatdelt{j}{\beta}{\mu}{C}$. For $\beta = 0$, note that $\newdelt{j}{0}{\mu}{C} = 0$, so define $\hatdelt{j}{0}{\mu}{C} = 0$; we will not need to probe the system at infinite temperature to produce this estimate. To compute the Taylor coefficient $\obsvalue{\beta^{(k)}}{t^{(j)}}{\mu}{C}$, we must estimate the coefficient of $\beta^k$ in the series expansion of $\newdelt{j}{\beta}{\mu}{C}$ in $\beta$. We once again use finite differencing, but this time we account for the truncation error using Theorem~\ref{thm:hightemp} and Lemma~\ref{lem:fewpolys} instead of Lemma~\ref{lem:hadamard}. We want to bound the approximation error
    \begin{multline}
        \Bigl|\obsvalue{\beta^{(k)}}{t^{(j)}}{\mu}{C} - \frac{1}{\beta^k k!} \sum^k_{\ell = 0} (-1)^{k-\ell} \binom{k}{\ell} \hatdelt{j}{\ell\beta}{\mu}{C} \Bigr| \\
        \le \Bigl|\obsvalue{\beta^{(k)}}{t^{(j)}}{\mu}{C} - \frac{1}{\beta^k k!} \sum^k_{\ell = 0} (-1)^{k-\ell} \binom{k}{\ell} \newdelt{j}{\ell\beta}{\mu}{C} \Bigr| + \frac{1}{k!}(2/\beta)^k\cdot \mathcal{O}(t\exp(\Od(j)))\,. \label{eq:betaderiv}
    \end{multline}
    Recall Eq.~\eqref{eq:series_t}, which in our notation can now be expressed as $\obsvalue{\beta}{t}{\mu}{C} = \sum^\infty_{j=0} \newdelt{j}{\beta}{\mu}{C} t^j$. By the reasoning in the proof of Lemma~\ref{lem:fewpolys}, the Taylor error in truncating $\newdelt{j}{\ell\beta}{\mu}{C}$ past the $(\ell\beta)^k$ term is bounded by $\sum^\infty_{k'=k+1} (\ell\beta)^{k'} q_{k'}$ where each $q_{k'}$ is a linear combination with $\exp(\mathcal{O}(j))$-bounded coefficients of at most $\exp(\Od(j))$ many polynomials of the form $p_{k'}$ from Theorem~\ref{thm:hightemp}, each of which has coefficients summing to $\exp(\Od(k))$. The Taylor error is thus at most $\sum^\infty_{k'=k+1} (\ell\beta)^{k'} \cdot \exp(\Od(j+k'))$, which is bounded by $(\ell\beta)^{k+1} \exp(\Od(j))$ provided $k\beta$ is bounded by some absolute constant $\beta_{\sf crit}$ depending only on the locality of $H$ and $\mathfrak{d}$. We can thus bound the first term on the right-hand side of Eq.~\eqref{eq:betaderiv} by $\beta \Od(k)^{k+1} \exp(\Od(j))$.

    Altogether, this requires estimating $\hatdelt{j}{\ell\beta}{\mu}{C}$ for $0\le\ell\le k$, each to additive error $\mathcal{O}(t\exp(\Od(j)))$, for a total of $(k+1)n_{\sf query}$ queries and total evolution time $j^2 k\, \Theta_{\mathfrak{d}}(t)^{-2j-1} \log 1/\delta$, and the aggregate error to which we can estimate $\obsvalue{\beta^{(k)}}{t^{(j)}}{\mu}{C}$ is 
    \begin{equation}
        \beta \Od(k)^{k+1} \exp(\Od(j)) + \frac{1}{k!}(2/\beta)^k\cdot \mathcal{O}(t\exp(\Od(j)))
    \end{equation}
    We can thus take $t = \nu\beta^k \cdot \exp(-\Theta_{\mathfrak{d}}(j+k))$, resulting in the claimed complexity bounds.
\end{proof}

\noindent In the last stage of our algorithm, we will also need to make use of the following truncation error bound for approximating derivatives of $\newdelt{j}{\beta}{\mu}{C}$. Write $C^\dagger[[H,\sigma^\mu_\probeindex]_j] = \sum_a c_a \sigma_a$, recalling from the proofs of Lemma~\ref{lem:fewpolys} and Lemma~\ref{lem:derivestimate} that the Paulis are $\mathcal{O}(j)$-local, there are at most $\exp(\Od(j))$ many, and their coefficients are $\exp(\mathcal{O}(j))$-bounded. Letting $p_{a;k}$ denote the polynomial $p_k$ arising in the series expansion of $\tr(\sigma_a \rho_\beta)$ in Theorem~\ref{thm:hightemp}, we thus have $\newdelt{j}{\beta}{\mu}{C} = \sum^\infty_{k=0} \beta^k \sum_a c_a  p_{a;k}$ and more generally
\begin{equation}
    \frac{1}{k!}\frac{\partial^k}{\partial \beta^k} \newdelt{j}{\beta}{\mu}{C} = \sum_{k'\ge k} \beta^{k'-k} \binom{k'}{k} \sum_a c_a p_{a;k'}\,.
\end{equation}

\begin{lemma}\label{lem:hideg}
    Let $\epsilon, h > 0$. Let $j,k\in\mathbb{Z}_{\ge 0}$, and define $\overline{k}$ to satisfy $\overline{k} - k = \Theta_{\mathfrak{d}}\Bigl(\frac{\log (1/\epsilon) + j + k}{\log 1/\beta}\Bigr)$ for sufficiently large constant factor. Let $\beta \le \beta_{\sf crit} - kh$. 
    
    Given probe access to $H$ at temperatures no higher than $1/\beta$, one can estimate 
    \begin{equation}
        \hideg{j}{k}{\overline{k}} \triangleq \sum_{k \le k' \le \overline{k}} \beta^{k'-k} \binom{k'}{k} \sum_a c_a p_{a;k'} = \obsvalue{\beta^{(k)}}{t^{(j)}}{\mu}{C} + \sum_{k < k' \le \overline{k}} \beta^{k'-k} \binom{k'}{k} \sum_a c_a p_{a;k'} \label{eq:extraterms}
    \end{equation}
    for any $\mu\in\{X,Y,Z\}$ and any single-qubit control channel $C$ to additive error $\epsilon + h\,\Od(k)^{k+1}\exp(\Od(j))$ with probability at least $1 - \delta$ using $\Od(1/h)^{2k(j+1)}\cdot \mathcal{O}(\log 1/\delta)$ queries with total evolution time $\Od(1/h)^{k(2j+1)}\cdot \mathcal{O}(\log 1/\delta)$.

    In addition, the polynomial in Eq.~\eqref{eq:extraterms}, regarded as a formal polynomial $\hideg{j}{k}{\overline{k}}(\params)$ in the unknown parameters $\params$ of $H$, can be computed in time $\Theta_{\mathfrak{d}}\Bigl(\frac{\log (1/\epsilon) + j + k}{\log 1/\beta}\Bigr)\cdot \exp(\Od(j+k))$.
\end{lemma}

\noindent When $j,k,\mathfrak{d} = \mathcal{O}(1)$, the corresponding degree-$(k + \Theta(\frac{\log 1/\epsilon}{\log 1/\beta}))$ truncated polynomial can be estimated to error $\epsilon + \mathcal{O}(h)$ in $\poly(1/h)$ queries and evolution time, and the polynomial can be written down as a formal polynomial in the unknown parameters of $H$ in time $\Theta(\frac{\log 1/\epsilon}{\log 1/\beta})$. Importantly, only the degree of the truncation, and not the complexity of estimation, depends on $\epsilon$.

\begin{proof}
    We can control the truncation error via:
    \begin{align}
        \Bigl|\frac{1}{k!}\frac{\partial^k}{\partial \beta^k} \newdelt{j}{\beta}{\mu}{C} - \sum_{k \le k'\le \overline{k}} \beta^{k'-k} \binom{k'}{k} \sum_a c_a p_{a;k'}\Bigr| &\le \sum_{k'>\overline{k}}\beta^{k' - k}\binom{k'}{k} \sum_a |c_a| \cdot |p_{a;k'}| \\
        &\le \exp(\Od(j+k)) \sum_{k'> \overline{k}} \Od(\beta)^{k'-k} \\
        &\le \exp(\Od(j+k))\cdot \Od(\beta)^{\overline{k} - k}\,. \label{eq:higherdegree1}
    \end{align}
    In particular, by the choice of $\overline{k}$ in the lemma, the above is at most $\epsilon$.
    
    As in the proof of Lemma~\ref{lem:derivestimate}, let $\hatdelt{j}{\beta}{\mu}{C}$, recalling this was guaranteed to be $\mathcal{O}(t\exp(\Od(j)))$. Using finite differencing, for any $h > 0$ we can consider
    \begin{multline}
        \Bigl|\frac{\partial^k}{\partial \beta^k} \newdelt{j}{\beta}{\mu}{C} - \frac{1}{h^k} \sum^k_{\ell = 0} (-1)^{k-\ell}\binom{k}{\ell}\hatdelt{j}{\beta + \ell h}{\mu}{C}\Bigr| \\
        \le \Bigl|\frac{\partial^k}{\partial \beta^k} \newdelt{j}{\beta}{\mu}{C} - \frac{1}{h^k}\sum^k_{\ell=0}(-1)^{k-\ell}\binom{k}{\ell}\newdelt{j}{\beta + \ell h}{\mu}{C}\Bigr| + \mathcal{O}(t\exp(\Od(j)))\cdot (2/h)^k\,. \label{eq:higherdegree2}
    \end{multline}
    The first term on the right-hand side is almost identical to the Taylor error bounded at the end of the proof of Lemma~\ref{lem:derivestimate}; indeed, by the same argument, one can bound this by $h \Od(k)^{k+1}\exp(\Od(j))$; note that this scales with $h$ instead of $\beta$ because we are evaluating the estimates $\widehat{\Delta}^{(j)}_{\mu,C}$ at $\beta, \beta + h, \ldots, \beta+kh$ rather than $0,\beta,\ldots,k\beta$. Combining this with Eqs.~\eqref{eq:higherdegree1} and~\eqref{eq:higherdegree2}, we conclude that using the estimates $\hatdelt{j}{(\ell+1)\beta}{\mu}{C}$ for $0 \le \ell \le k$, we can estimate $\sum_{k\le k'\le \overline{k}}\beta^{k'-k}\binom{k'}{k} \sum_a c_a p_{a;k'}$ to error
    \begin{equation}
        \epsilon + h\,\Od(k)^{k+1}\exp(\Od(j)) + \mathcal{O}(t\exp(\Od(j)))\cdot (2/h)^k\,.
    \end{equation}
    In particular, if we take $t = h^{k+1}$, then the third term is dominated by the second, yielding the desired error bound. Substituting this choice of $t$ into the complexity bounds in Lemma~\ref{lem:deltestimate} for producing the estimates $\hatdelt{j}{\beta}{\mu}{C}$ yields the claimed complexity bound.
\end{proof}

The sequence of lemmas in this Section establishes that the Taylor coefficients of our probe observables are bounded-degree polynomials in the Hamiltonian parameters with controlled coefficients, and that they can be estimated to prescribed accuracy using a finite number of queries with bounded evolution time. Together, these results furnish the bridge between experimental data and the algebraic systems that underlie our convergence and identifiability analyses in subsequent sections.

\section{Certifying Generic Fibers for Rectangular Maps}
\label{sec:verify-generic-fiber}

Our goal in this section is to address the following question.  Suppose we have a set of $m$ polynomials in $N$ variables where $m \geq N$.  Now suppose we consider a generic input in $\mathbb{R}^N$ and evaluate the polynomials, and then take the inverse image of the result.  The original point will be in that inverse image (the fiber), but other points may be as well, possibly including complex-valued ones.  Here we show, under certain conditions, that generically the size of these fibers is constant.  This will enter into our analysis in the following way.  Our $m$ probe interrogations of a thermal system  result in $m$ numbers, which we can regard as the right-hand side of a polynomial system in $N$ variables describing the Hamiltonian.  Then we want to reconstruct the variables of that Hamiltonian, and it is useful to know under generic conditions how many `solutions' to the Hamiltonian parameters we will find given our measurements.  Among those solutions will be the `true' Hamiltonian parameters, but there may be others.  We will choose measurements so that the solutions can be classified and interpreted.

With this motivation, we turn to the problem of verifying the size of `generic' fibers of a polynomial system, under sensible conditions satisfied by our physical example.  To this end, consider a set of $m$ polynomials with rational coefficients, namely
\begin{align}
P=(p_0,\ldots,p_{m-1})\subset \mathbb{Q}[X_1,\ldots,X_N],\quad m\ge N\,.
\end{align}
We view this as a polynomial map $P=(p_0,\ldots,p_{m-1}):\mathbb{C}^N\to\mathbb{C}^m$ and define its image variety as the Zariski closure $\mathbb{Y}=\overline{P(\mathbb{C}^N)}^{\mathrm{Zar}}\subset\mathbb{C}^m$. More formally, our goal is to certify the generic degree
\begin{align}
g := \text{gdeg}\bigl(P:\mathbb{C}^N\to \mathbb{Y}\bigr),
\end{align}
i.e.~the number of points in a generic fiber $P^{-1}(\bc)$ for $\bc \in \mathbb{Y}$.  Moreover, for a real-valued right-hand side $\bc \in P(\R^N)\subset \mathbb{Y}\cap\R^m$, we want to verify that the complex fiber $P^{-1}(\bc)\subset\mathbb{C}^N$ has size $g$ for generic choices of $\bc$ in the real image. We emphasize that the fiber may include solutions with complex coordinates, and our certificate will count the total number of complex points.

We let $J_P(x)$ denote the $m \times N$ Jacobian matrix of $P$ and call an $N\times N$ minor of $J_P$ a Jacobian minor.  In the special case that $m = N$, the unique $N \times N$ Jacobian minor is just the determinant of the Jacobian.  We have the following lemma and theorem.

\begin{lemma}\label{lem:rank-dim}
If some  $N \times N$ Jacobian minor of  $J_P$ is not the zero polynomial, then $\dim \,\mathbb{Y} = N$.
\end{lemma}

\begin{proof}
Choose $\textbf{x}_0\in\mathbb{C}^N$ where that minor is nonzero. Then the differential $dP_{\textbf{x}_0}$ has rank $N$. By the holomorphic implicit function theorem (in constant-rank normal form), in neighborhoods of $\textbf{x}_0$ and $\textbf{c}_0:=P(\textbf{x}_0)$ there are holomorphic coordinates in which $P$ is a submersion onto an $N$–dimensional complex submanifold through $\textbf{c}_0$. Hence the image has local (and therefore global) dimension at least $N$ at $\textbf{c}_0$, i.e.~$\dim \,\mathbb{Y}\geq N$. Since the domain has dimension $N$, we also have $\dim \,\mathbb{Y}\leq N$. Therefore $\dim \,\mathbb{Y} = N$.
\end{proof}

\begin{theorem}[Finite \'{e}tale cover of the image]\label{thm:finiteetale}
Under the hypothesis of Lemma~\ref{lem:rank-dim}, the polynomial system $\{P = 0\}$ satisfies generic finiteness in the sense of Definition~\ref{def:finiteness}, and furthermore the fibers consist of \emph{non-singular (isolated)} points.
% Then there exists an integer
% $g \geq 1$ and a nonempty Zariski-open dense subset $U \subset \mathbb{Y}$ such that
% \begin{align}
% P^{-1}(U)\ \xrightarrow{ P }\ U
% \end{align}
% is a finite \'{e}tale cover of degree $g$. Equivalently, for every $\textbf{c}\in U$, the fiber $P^{-1}(\textbf{c})$ consists of exactly $g$ distinct nonsingular (isolated) points. In particular, $g$ is the generic degree of $P$.
\end{theorem}

\begin{proof}
By Lemma~\ref{lem:rank-dim}, $\dim \,\mathbb{Y} = N$. By the fiber-dimension theorem, the minimal fiber dimension of $P:\mathbb{C}^N\to \mathbb{Y}$ is $N-\dim \,\mathbb{Y}=0$, so there exists a nonempty Zariski-open $V \subset \mathbb{Y}$ over which all fibers are zero-dimensional and hence finite. Thus $P$ is \emph{quasi-finite} over $V$ (i.e.~of finite type with finite fibers).  By Zariski's Main Theorem, after shrinking to a smaller nonempty open $U_1\subset V$ the map $P^{-1}(U_1)\to U_1$ is \emph{finite} (i.e.~affine with coordinate rings finite over the base; in particular, proper with finite fibers).  In characteristic~$0$, deleting the branch locus from $U_1$ yields a nonempty open $U \subset U_1$ over which $P$ is \'{e}tale (equivalently, smooth of relative dimension~$0$). Since $P$ is already finite on $U_1$, it follows that $P^{-1}(U)\to U$ is finite \'{e}tale (i.e.~finite and smooth of relative dimension $0$, or flat and unramified).

Because $\mathbb{Y}$ is irreducible, any nonempty Zariski-open $U\subset \mathbb{Y}$ is Zariski-dense (and hence dense in the Euclidean topology as well). Over the connected base $U$, a finite \'{e}tale cover has constant degree; denote this degree by $g$. By definition, $g$ is the generic fiber cardinality of $P$, so each $\textbf{c}\in U$ has $|P^{-1}(\textbf{c})|=g$ reduced points. This $g$ is the generic degree.
\end{proof}

\begin{remark}[Certifying the generic fiber size]\label{remark:certify_generic}
Assume some $N \times N$ Jacobian minor of $J_P$ is not identically zero. 
Pick a point $\textbf{x}_{\,0} \in \mathbb{C}^N$ where that minor is nonzero and set 
$\textbf{c}_{\,0} = P(\textbf{x}_{\,0})$ and solve $P(\textbf{x}) = \textbf{c}_{\,0}$.  If $P(\textbf{x}) = \textbf{c}_0$ has $g < \infty$ solutions which are isolated and have the maximal Jacobian rank $N$, then  $|P^{-1}(\textbf{c}_{\,0})|$ equals the generic degree $g$. 

Since the branch locus is measure zero in $\mathbb{Y}$, suitable $\textbf{c}_{\,0}$'s (and hence suitable $\textbf{x}_{\,0}$'s) are generic, and hence should be easy to find.

Moreover, the branch locus is a proper Zariski-closed subset; hence its complement $U$ is nonempty Zariski-open. Since $P(\mathbb{R}^N)$ is Zariski-dense in $\mathbb{Y}$, we can find a suitable $\textbf{c}_{\,0}$ from a generic $\textbf{x}_{\,0} \in \mathbb{R}^N$ as opposed to $\mathbb{C}^N$, if we desire.
\end{remark}

As part of certifying the generic fiber size, we need to solve a polynomial system with isolated, non-degenerate solutions.  We note that there are many methods for doing this symbolically (see e.g.~\cite{cox1998using} for an overview), as well as numerically to any specified precision (see e.g.~\cite{sommese2005numerical, burgisser2013condition}; for a more recent discussion with useful statements about computational complexity see~\cite{el2018bit}).

\section{Quantitative Identifiability for Non-Degenerate Systems}
\label{sec:quantitative}

In this section, we show that under certain non-degeneracy conditions on the polynomial map $P$ from Section~\ref{sec:verify-generic-fiber}, given any point which is sufficiently close to a point  $P(\bx^*)$ in the image of $P$, then it must be close to $\bx^*$, or more precisely, to $\bx^*$ up to simple and unavoidable symmetries.

\subsection{Setup}

\paragraph{Notation and smoothed model.} Here we use similar notation as Section~\ref{sec:verify-generic-fiber}. Let $P = (p_0,\ldots,p_N): \C^N \to \C^{N+1}$ be a polynomial map with image variety $\mathbb{Y}$ and variables denoted by $X_1,\ldots,X_N$. Letting $L$ denote the linear projection onto the last $N$ coordinates of $\R^{N+1}$ and $L^\perp$ the projection onto the first coordinate, define $F = L\circ P: \C^N \to \C^N$ and $G = L^\perp \circ P: \C^N\to \C$. We will denote the Jacobian of $F, P$ by $J_F, J_P$ respectively, and the Hessian of $p_i$ by $\Hess_i$. 

As before, we will use $\bc\in \mathbb{Y}$ to denote the ``right-hand side'' of the polynomial system. In this section, this right-hand side will be generated via the following probabilistic process, which is directly motivated by our model for smoothed Hamiltonians in Definition~\ref{def:smoothed_analysis}:
\begin{enumerate}
    \item Nature samples $\bx^* \sim \mathcal{N}(\mu,\smoothing^2\, \Id)$, where $\mu\in\C^N$ is an unknown parameter and $\smoothing$ is the smoothing parameter in the sense of Definition~\ref{def:smoothed_analysis}.
    \item $\bc \triangleq P(\bx^*)$. For convenience, define $c_0 \triangleq L^\perp \bc \in \C$ and $\bc' = (c_1,\ldots,c_N) \triangleq L\bc$.
\end{enumerate}
This can be thought of as a quantitative notion of genericity.

\paragraph{Symmetries in common zero set.} The common zeros of the polynomial systems we consider in this work naturally have symmetries.

\begin{assumption}\label{assume:symmetry}
    $\C^N$ is equipped with an action of the group $\mathcal{G} = \mathcal{S}_\ell$ such that for every $\pi\in \mathcal{G}$, we have $P(\bx) = P(\pi\cdot \bx)$, and furthermore the set of points in $\C^N$ belonging to an orbit of size less than $\ell$ is of Lebesgue measure zero. 
\end{assumption}

\paragraph{Generically finite fibers.} Recall the notation of Section~\ref{sec:alggeo}. We will make the following assumptions about $P$ which were also leveraged in Section~\ref{sec:verify-generic-fiber}. 

\begin{assumption}\label{assume:fiber}
    The (square) polynomial system $\{p_1 - c_1 = \cdots = p_N - c_N = 0\}$ satisfies generic finiteness (Definition~\ref{def:finiteness}). Moreover, the size of a generic affine fiber of the (rectangular) system $\{p_0 - c_0 = \cdots = p_N - c_N = 0\}$ is $\ell$. In particular, by Assumption~\ref{assume:symmetry}, for a generic point in $Y$, its fiber under $P$ corresponds to a single orbit under the action of $\mathcal{G}$.
\end{assumption}
\noindent The first part of this assumption will allow us to exploit the product formula (Lemma~\ref{lem:poisson}. The second part of this assumption implies that almost surely over the randomness of $\bx^*$, we have that $\zeroset$ is the orbit of $\bx^*$.

\paragraph{Radius bounds.} Finally, we will use a basic bound on how large a fixed polynomial can be over the ground truth point $\bx^*$.

\begin{lemma}\label{lem:basic_conc}
    For any polynomial $f$ whose coefficients only depend on $F$ or $P$, there is a $\tau_f = \poly_P(\norm{\mu}_\infty,\smoothing, \log 1/\delta)$ for which
    \begin{equation}
        \Pr{\norm{f(\bx^*)}_\infty \ge \tau_f} \le \delta\,.
    \end{equation}
\end{lemma}

\begin{proof}
    This immediately follows from the fact that $\norm{\bx^*}_\infty \le \norm{\mu}_\infty + 2\smoothing\sqrt{\log N/\delta}$ with probability at least $1 - \delta$ by standard Gaussian concentration.
\end{proof}

\noindent We will use the bound in Lemma~\ref{lem:basic_conc} often enough that it will be convenient to define the notation 
\begin{equation}
    \calR \triangleq \poly_P(\norm{\mu}_\infty,\smoothing,\log 1/\delta)\,.
\end{equation}
Under this shorthand, $\calR^{\mathcal{O}(1)}$ is ``equal'' to $\calR$, and any polynomial function of $\calR$ that only depends on the coefficients of $P$ is also ``equal'' to $\calR$. This will allow us to sidestep the cumbersome exercise of tracking specific constants as they depend on coefficients of $P$.

\paragraph{Proof sketch.} Here we outline our general strategy. We will first show in Section~\ref{sec:avoid} that the constraint associated to the polynomial $G$ is violated by all points in $\zeroset'\backslash\zeroset$, i.e., by all points which satisfy the constraints associated to $F$ but which are not among the points in the orbit of $\bx^*$. Then in Section~\ref{sec:neighborhood}, we show that the only points which \emph{approximately} satisfy the constraints associated to $F$ lie in a neighborhood of $\zeroset'$. Such points can be obtained by, e.g., exhaustively enumerating over a fine grid of parameter space and only keeping those points which approximately satisfy the constraints associated to $F$. Finally, in Section~\ref{sec:ruleout}, we use the fact proved in Section~\ref{sec:avoid} that the constraint associated to $G$ is violated by points in $\zeroset'\backslash\zeroset$, together with Lipschitzness of $G$, to conclude that all of the spurious points found by exhaustive enumeration, i.e., ones not close to the orbit of $\bx^*$, will fail to approximately satisfy the constraint associated to $G$ as expected. This will allow us to whittle down the points in the neighborhood of $\zeroset'$ to only those in a neighborhood of $\zeroset$. Our final algorithm {\sc FindRoot} (Algorithm~\ref{alg:findroot}), which we present in Section~\ref{sec:puttogether}, then takes these points and runs a few Newton steps to push them even closer to $\zeroset$ (see Theorem~\ref{thm:solvesystem} for the final guarantee).

As we progress through the steps of the proof, we will pick up some additional assumptions that have to be made --- and which are efficiently certifiable --- to ensure $P$ is sufficiently ``non-degenerate'' for the main result of this section to hold. While the reader might be worried about the length of this list of assumptions (six in total), it turns out that all of them can be certified by simply running the check in Remark~\ref{remark:certify_generic} - see Remark~\ref{remark:check} at the end for further discussion.

\subsection{Avoiding Near Zeroes}
\label{sec:avoid}

In this section we show that with high probability over $\bx^*$ generated via the probabilistic process above, for every $\bx\in \zeroset'\backslash \zeroset$ the quantity $|G(\bx) - c_0|$ is not too small. For this, we will need one more assumption that requires some development. This is the first of many places where we will make use of the product formula (Lemma~\ref{lem:poisson}).

\begin{lemma}\label{lem:apply_poisson_product}
    The product $\prod_{\bx\in \zeroset'}(c_0 - G(\bx))$, regarded as a function in $\bc$, is equal to $R(\bc)^\ell$ for a polynomial $R$ in the variables $\bc$.
\end{lemma}

\begin{proof}
    By the product formula (Lemma~\ref{lem:poisson}) - with $q_1,\ldots,q_N$ therein taken to be $p_1 - c_1,\ldots,p_N - c_N$, $f$ taken to be $G - c_0$, and specializing $t$ therein to $0$ - we find that $R'(\bc) \triangleq \prod_{\bx\in \zeroset'} (c_0 - G(\bx))$ can be regarded formally as a polynomial in the variables $\bc$. Because the set of $\bx \in \C^N$ whose orbit under the action of $\mathbb{S}_\ell$ has size less than $\ell$ is of measure zero, the factors of $\prod_{\bx \in \zeroset'} (c_0 - G(\bx))$ come in groups of $\ell$, each corresponding to a full orbit. Denote these factors by $c_0 - z_1(\bc'), \ldots, c_0 - z_m(\bc')$ for $m \triangleq |\zeroset'|/\ell$. Therefore, $R'(\bc) = R(\bc)^\ell$ for polynomial $R(\bc) = \prod^m_{i=1} (c_0 - z_i(\bc'))$ as claimed.
\end{proof}

\noindent Regarding $R(\bc)$ as a polynomial in the single variable $c_0$, let $\mathfrak{D}(\bc')$ denote its discriminant. In the notation of the proof above,
\begin{equation}
    \mathfrak{D}(\bc') = \prod_{i<j} (z_i(\bc') - z_j(\bc'))^2\,. \label{eq:discriminant}
\end{equation}
Note that the coefficients of $\mathfrak{D}$ only depend on the coefficients of $P$ and, in particular, are independent of $\bc$. The following assumption is a quantitative version of the assumption that the discriminant is not too small when evaluated at $c_0$, which will be essential for showing that $|G(\bx) - c_0|$ is not too small for all $\bx\in \zeroset'\backslash\zeroset$ in Section~\ref{sec:avoid}.

\begin{assumption}\label{assume:hermite_discriminant}
    Let $\mathfrak{D}\circ F$ have Hermite expansion
    \begin{equation}
        \mathfrak{D}\circ F = \sum_\alpha b_\alpha h_\alpha\,.
    \end{equation}
    Then $\sum_{\alpha \ \mathrm{maximal}} b_\alpha^2 \ge \gamma^2$ for some parameter $\gamma > 0$.
\end{assumption}

\noindent Pragmatically, this can be verified by checking, for a random choice of $\bx^*$, whether the associated $\zeroset$ has the correct size $\ell$. Under this assumption, we can use polynomial anti-concentration applied to $\mathfrak{D}\circ F$ to conclude that the $z_i(\bc')$'s are well-separated.

\begin{lemma}\label{lem:avoidzeros}
    With probability at least $1 - \delta$ over $\bx^*$,
    \begin{equation}
        \min_{\bx \in \zeroset'\backslash \zeroset} |G(\bx) - c_0| \ge \gamma\,\poly_P(\delta,\smoothing) / \calR \,.
    \end{equation}
\end{lemma}

\begin{proof}
    To bound $\min_{\bx\in\zeroset'\backslash\zeroset}|G(\bx) - c_0|$, it suffices to bound $\min_{i<j} |z_i(\bc') - z_j(\bc')|$, in the notation of the proof of Lemma~\ref{lem:apply_poisson_product}. To do that, we will leverage Eq.~\eqref{eq:discriminant}. Recalling that $\bc' = F(\bx^*)$ for $\bx^*\sim \mathcal{N}(\mu,\smoothing^2\cdot\Id)$ and noting that $\mathfrak{D}\circ F$ is some polynomial of degree $O_P(1)$, by Carbery-Wright we have that
    \begin{equation}
        \Pr{|\mathfrak{D}(\bc')| \le \mathrm{poly}_F(\delta)\cdot\Var(\mathfrak{D}\circ F)^{1/2}} \le \delta\,.
    \end{equation}
    Lemma~\ref{lem:anticonc} and Assumption~\ref{assume:hermite_discriminant} imply that $\Var(\mathfrak{D}\circ F) \ge \poly_P(\smoothing)\gamma^2$, so with probability at least $1 - \delta$, $|\mathfrak{D}(\bc')| > \gamma\,\poly_P(\delta,\smoothing)$.

    Recall that the $z_i(\bc')$'s are roots of the polynomial $\prod_{\bx\in \zeroset'} (c_0 - G(\bx)) = \sum^r_{t=0} a_t(\bc') c_0^t$, where $a_0,\ldots,a_r$ are polynomials whose coefficients only depend on $F$ and $a_r = 1$. By Cauchy's root bound, $|z_i(\bc')| \le 1 + \max_{t<r} |a_t(\bc')|$, so by Lemma~\ref{lem:basic_conc} applied to $a_t\circ F$, we conclude that $|z_i(\bc') - z_j(\bc')| \le \calR$ for all $i,j$. From this and our lower bound on $|\mathfrak{D}(\bc')|$, our claimed bound follows as there are $\binom{m}{2} = \binom{|\zeroset'|/\ell}{2} = O_P(1)$ pairs $1\le i<j < m$. Note that technically we need to apply a union bound, but it is over $O_P(1)$ many things, and replacing $\delta$ in $\calR$ by $\delta/O_P(1)$ does not change the asymptotic scaling of $\calR$.
\end{proof}

\subsection{Finding Neighborhood of \texorpdfstring{$\zeroset'$}{V'}}
\label{sec:neighborhood}

In this section, we show that if some candidate solution $\bx$ approximately satisfies the constraints associated to $F$, then it must lie in a small neighborhood around one of the points in $\zeroset'$. To do this, we would like to use the quantitative inverse function theorem (Lemma~\ref{lem:inverse}), but for that we need to ensure that the Jacobian of $F$ at each of the points in the fiber $F^{-1}(\bc')$ is well-conditioned. Note that this is more subtle than simply showing that $J_F(\bx^*)$ is well-conditioned, which is more straightforward as one can directly use anti-concentration of the determinant of $J_F$ at $\bx^*$, which is a simple application of the tools from Section~\ref{sec:anticoncentration}. Instead, we need to argue that with high probability over $\bx^*$, \emph{all} the points in $F^{-1}(F(\bx^*))$ have well-conditioned Jacobian. Once again, the product formula will play a key role, and indeed the argument is very similar to the one in the previous section.

As in Section~\ref{sec:verify-generic-fiber}, let $\jac_F \triangleq \det J_F \in \C[X_1,\ldots,X_N]$. Let $\jac^{\sf hom}_F \in \C[X_0,\ldots,X_N]$ denote its homogenization.

\begin{lemma}
    The product $\prod_{\bx\in \zeroset'} \jac_F(\bx)$, regarded as a function of $\bc'$, is equal to a polynomial $\mathfrak{J}$ in the variables $\bc'$.
\end{lemma}

\begin{proof}
    This follows immediately by the product formula (Lemma~\ref{lem:poisson}), this time with $f$ in Lemma~\ref{lem:poisson} taken to be $\jac_F$.
\end{proof}

\noindent The following assumption is a quantitative version of the assumption that $\bc'$ avoids the zero locus of $\mathfrak{J}$, so that $J_F$ at each of the points in the fiber $F^{-1}(\bc')$ is well-conditioned. Pragmatically, it can be verified by checking whether, for a randomly sampled $\bx^*$, the fiber $F^{-1}(F(\bx^*))$ has nonsingular Jacobians.

\begin{assumption}\label{assume:jacprod}
    Let $\mathfrak{J}\circ F$ have Hermite expansion
    \begin{equation}
        \mathfrak{J}\circ F = \sum_\alpha b'_\alpha h_\alpha\,.
    \end{equation}
    Then $\sum_{\alpha \ \mathrm{maximal}} b^2_\alpha \ge \gamma'^2$ for some parameter $\gamma' > 0$.
\end{assumption}

\begin{lemma}\label{lem:sigmin_lbd}
    With probability at least $1 - \delta$ over $\bx^*$,
    \begin{equation}
        \min_{\bx\in \zeroset'} \sigma_{\min}(J_F(\bx)) \ge \gamma'\,\poly_P(\delta,\smoothing) / \calR\,.
    \end{equation}
\end{lemma}

\noindent The proof is very similar to that of Lemma~\ref{lem:avoidzeros}, except that we need to use slightly different resultants. We include it for completeness.

\begin{proof}
    Noting that $\mathfrak{J}\circ F$ is some polynomial of degree $O_P(1)$, by Carbery-Wright we have that
    \begin{equation}
        \Pr{|\mathfrak{J}(\bc')| \le \poly_P(\delta)\cdot \Var(\mathfrak{J}\circ F)^{1/2}} \le \delta\,.
    \end{equation}
    Lemma~\ref{lem:anticonc} and Assumption~\ref{assume:jacprod} imply that $\Var(\mathfrak{J}\circ F) \ge \poly_P(\smoothing)\gamma'^2$, so with probability at least $1 - \delta$, $\mathfrak{J}(\bc')| > \gamma'\,\poly_P(\delta,\smoothing)$. Combining this with the upper bound on $\jac_F(\bx)$ for $\bx\in\zeroset'$ in Corollary~\ref{cor:jacbounds} below, we conclude that with probability at least $1 - \delta$,
    \begin{equation}
        \min_{\bx\in \zeroset'} \jac_F(\bx) \ge \gamma'\,\poly_P(\delta,\smoothing) / \calR\,.
    \end{equation}
    Because $\jac_F \le \norm{J_F}^{N-1}_{\sf op} \cdot \sigma_{\min}(J_F)$, the claim follows by the upper bound on $\norm{J_F}_F$ in Corollary~\ref{cor:jacbounds} below.
\end{proof}

\noindent The proof above required the following upper bounds.

\begin{lemma}\label{lem:upperbounds}
    With probability at least $1 - \delta$ over $\bx^*$,
    \begin{equation}
        \max_{\bx\in\zeroset'} \norm{\bx}_\infty \le \calR\,.
    \end{equation}
\end{lemma}

\begin{proof}
    By the product formula (Lemma~\ref{lem:poisson}), this time applied to $f = X_i$, for any $i\in[N]$ we have
    that the map $(t,\bc') \mapsto \prod_{\bx\in\zeroset'} (t - \bx_i)$ is a polynomial whose coefficients depend solely on the coefficients of $F$. Regarding this as a polynomial $\sum^r_{s=0} a_s(\bc') t^s$, Cauchy's root bound ensures that the roots $\{\bx_i\}_{\bx\in\zeroset'}$ are bounded in magnitude by $1 +  \max_{s<r}|a_s(\bc')|$, so by Lemma~\ref{lem:basic_conc} applied to $a_s\circ F$, we find by taking all $i\in[N]$ that $\norm{\bx}_\infty \le \calR$ for all $\bx\in\zeroset'$.
\end{proof}

\begin{corollary}\label{cor:jacbounds}
    For all $\bx\in\zeroset'$ and $j\in[N+1]$, $\norm{J_P(\bx)}^2_F, |\jac_F(\bx)|, \norm{ \Hess_j(\bx)}^2_F \le \calR$.
\end{corollary}

\begin{proof}
    This simply follows by the fact that the coefficients of $\jac_F$, $\norm{J_P}^2_F$, $\norm{\Hess_j(\bx)}^2_F$ as polynomials in $\bx$ are at most $O_P(1)$ in magnitude by definition, and the entries of $\bx$ are at most $\calR$ in magnitude by Lemma~\ref{lem:upperbounds}.
\end{proof}

\noindent Lastly, we will need to lower bound the distance from $\bc'$ to the discriminant locus. This will ensure that we do not run into points in the image whose fiber is infinitely large. Indeed, by the first part of our Assumption~\ref{assume:fiber} a generic fiber in the image variety of $F$ consists of exactly $d$ points. By B\'ezout, $d \le \prod^N_{i=1}\mathrm{deg} f_i = O_P(1)$, and every point in such a fiber is nonsingular. Furthermore, there is a nonzero polynomial $\Delta$ on the target such that for all points over which $\Delta$ does not vanish, the fiber satisfies these properties.

We need a quantitative bound on the extent to which $\Delta$ is nonzero. As explained in Section~\ref{sec:verify-generic-fiber}, this can be certified by again considering a random $\bx^*$ and verifying that the size of the fiber $F^{-1}(F(\bx^*))$ is finite.

\begin{assumption}\label{assume:Delta}
    Let $\Delta\circ F$ have Hermite expansion
    \begin{equation}
        \Delta\circ F = \sum_\alpha b''_\alpha h_\alpha\,.
    \end{equation}
    Then $\sum_{\alpha \ \mathrm{maximal}}b^2_\alpha \ge \gamma''^2$ for some parameter $\gamma'' > 0$.
\end{assumption}

\begin{lemma}\label{lem:dist_to_locus}
    There is a $\tau = \gamma''\,\poly_P(\delta,\smoothing)/\calR$ such that with probability at least $1 - \delta$, all $\bx\in B_\tau(\bc')$ lie outside the discriminant locus $\{\Delta = 0\}$.
\end{lemma}

\begin{proof}
    We first show that the value of $\Delta$ at $\bc'$ is lower bounded with high probability. As $\Delta$ only depends on $F$ and is a nonzero polynomial, by Carbery-Wright and the fact that $\Var(\Delta\cdot F) \ge \poly_P(\smoothing)\gamma''^2$ by Lemma~\ref{lem:anticonc}, we have
    \begin{equation}
        \Pr{|\Delta(\bc')| \le \gamma''\,\poly_P(\delta,\smoothing)} \le \delta\,.
    \end{equation}
    Furthermore, $\Delta$ is $\calR$-Lipschitz in a small neighborhood around $\bx^*$ with probability at least $1 - \delta$. In this case, in order for $\Delta$ to vanish at a point, that point must be at least $\tau$-far from $\bc'$.
\end{proof}

\noindent Finally, we are ready to lower bound the distance between the points in $\zeroset'$. We will use a resultant similar to the one in Lemma~\ref{lem:upperbounds}. For a fixed unit vector $\vphi\in\C^N$, by the product formula (Lemma~\ref{lem:poisson}) applied to $f = \langle \vphi, \bx\rangle$, the following is a polynomial in the variables $(t,\bc')$ whose coefficients depend solely on the coefficients of $F$.
\begin{equation}
    \prod_{\bx\in\zeroset'}(t - \langle \vphi, \bx\rangle)\,.
\end{equation}
Let $\mathfrak{S}_\vphi$ denote the discriminant of this as a polynomial in $t$, and regard the discriminant as a polynomial in $\bc'$ to obtain
\begin{equation}
    \mathfrak{S}_\vphi(\bc') = \pm\prod_{\bx,\bx'\in\zeroset': \bx\neq \bx'} \langle \vphi, \bx - \bx'\rangle\,. \label{eq:discriminantS}
\end{equation}

\begin{assumption}\label{assume:rootsep}
    Let $\mathfrak{S}_\vphi\circ F$ have Hermite expansion
    \begin{equation}
        \mathfrak{S}_\vphi\circ F = \sum_\alpha b^{\vphi}_\alpha h_\alpha\,.
    \end{equation}
    Then $\sum_{\alpha \ \mathrm{maximal}} (b^\vphi_\alpha)^2 \ge \gamma'''^2$ for some parameter $\gamma''' > 0$ and some unit vector $\vphi\in\C^N$.
\end{assumption}

\begin{lemma}\label{lem:rootsep}
    With probability at least $1 - \delta$,
    \begin{equation}
        \min_{\bx,\bx'\in \zeroset': \bx\neq \bx'} \norm{\bx - \bx'}_2 \ge \gamma'''\,\poly_P(\delta,\smoothing)/\calR\,.
    \end{equation}
\end{lemma}

\begin{proof}
    The coefficients of the discriminant in Eq.~\eqref{eq:discriminantS} as a polynomial in $\bc'$ only depend on $F$. So by Carbery-Wright, Lemma~\ref{lem:anticonc}, and Assumption~\ref{assume:rootsep},
    \begin{equation}
        \Pr{|\mathfrak{S}_\vphi(\bc')| > \gamma'''\,\poly_P(\delta,\smoothing)} \le \delta\,.
    \end{equation}
    On the other hand, all of the factors in the discriminant are upper bounded by $\calR$ by Lemma~\ref{lem:upperbounds}, so the claim follows.
\end{proof}

We are now ready to prove the main claim of this section, namely that a sufficiently accurate approximate solution to the system $F(\bx) = \bc$ must be close to one of the finitely many points in $F^{-1}(\bc')$.

\begin{lemma}\label{lem:square_identifiable}
    Let $\eta>0$ be any parameter satisfying $\eta \le \min(\gamma',\gamma'',\gamma''')\,\poly_P(\delta,\smoothing)/\calR$. Suppose that some $\hat{\bx}\in\C^N$ satisfies $\norm{F(\hat{\bx}) - F(\bx^*)}_2 \le \eta\gamma'\,\poly_P(\delta,\smoothing)/\calR$. Then there is some $\bx\in \zeroset'$ for which $\norm{\hat{\bx} - \bx} \le \eta$.
\end{lemma}

\begin{proof}
    By the quantitative inverse function theorem (Lemma~\ref{lem:inverse}), if we take $F$ and $\bx$ therein to be our $F$ and a point $\bx\in F^{-1}(\bc')$, and note that $\sigma_{\min}(J_F(\bx)) \ge \gamma'\,\poly_P(\delta,\smoothing)/\calR$ (by Lemma~\ref{lem:sigmin_lbd}) and $\smoothness = \calR$ (by Corollary~\ref{cor:jacbounds}), we conclude that for any $\eta < \gamma'\,\poly_P(\delta,\smoothing)/\calR$, for each such $\bx \in F^{-1}(\bc')$ there is a unique $\tilde{x}\in B_\eta(\bx)$ for which $F(\tilde{\bx}) = F(\hat{\bx})$ provided $\norm{F(\hat{\bx}) - F(\bx^*)}_2 \le \eta\gamma'\,\poly_P(\delta,\smoothing)/\calR$. We wish to take $\eta$ small enough that (1) these balls $B_\eta(\bx)$ for different $\bx\in F^{-1}(\bc)$ are disjoint, and (2) $B_\eta(\bx^*)$ does not intersect the discriminant locus $\{\Delta = 0\}$. (1) requires $\eta < \gamma'''\,\poly_P(\delta,\smoothing)/\calR$ and ensures that $F^{-1}(F(\hat{\bx}))$ contains at least $d$ many distinct points, where $d$ is the size of a generic fiber of $F$. (2) requires $\eta < \gamma''\,\poly_P(\delta,\smoothing) / \calR$ and ensures that there are only finitely many, and in fact at most $d$, distinct points in $F^{-1}(F(\hat{\bx}))$. Therefore, $\hat{\bx}$ must be one of these $d$ points, each of which we have ensured above is in an $\eta$-ball of some point $\bx\in \zeroset'$.
\end{proof}

\subsection{Ruling Out Spurious Zeros}
\label{sec:ruleout}

Finally, in this section we show that $G$ can be used to rule out the spurious zeros, i.e., the points in $\zeroset'\backslash \zeroset$.

\begin{theorem}\label{thm:ruleout}
    Let $\eta>0$ be any parameter satisfying $\eta \le \min(\gamma,\gamma',\gamma'',\gamma''')\,\poly_P(\delta,\smoothing)/\calR$. Suppose that some $\hat{\bx}\in\C^N$ satisfies $\norm{F(\hat{\bx}) - F(\bx^*)}_2 \le \eta\gamma'\,\poly_P(\delta,\smoothing)/\calR$ and $|G(\hat{\bx}) - c_0| \le \gamma\,\poly_P(\delta,\smoothing)/\calR$. Then there is some $\bx\in \zeroset$ for which $\norm{\hat{\bx} - \bx} \le \eta$.
\end{theorem}

\begin{proof}
    By Lemma~\ref{lem:square_identifiable}, the condition that $\norm{F(\hat{\bx}) - F(\bx^*)}_2 \le \eta\gamma'\,\poly_P(\delta,\smoothing)/\calR$ already implies that $\hat{\bx}$ lives in an $\eta$-ball around some point $\bx\in\zeroset'$. Suppose to the contrary that $|G(\hat{\bx}) - c_0| \le \gamma$ yet $\bx\not\in\zeroset$. By Lemma~\ref{lem:avoidzeros} we know that $|G(\bx) - c_0| \ge \gamma\,\poly_P(\delta,\smoothing)/\calR$. But because $\norm{\bx}_\infty \le \calR$ by Lemma~\ref{lem:upperbounds} and $G$ is a polynomial whose coefficients and degree are bounded by $O_P(1)$ by definition, $G$ is $\calR$-Lipschitz in a neighborhood around $\bx$, meaning that $|G(\hat{\bx}) - c_0| \ge \gamma\,\poly_P(\delta,\smoothing)/\calR - \eta\calR$. As we assumed $\eta \le \gamma\,\poly_P(\delta,\smoothing)/\calR$, this would contradict the assumption that $|G(\hat{\bx}) - c_0| \le \gamma\,\poly_P(\delta,\smoothing)/\calR$.
\end{proof}

\subsection{Proof of Main Identifiability Result}
\label{sec:puttogether}

We are now ready to state and prove our main theorem, namely that under the above assumptions on $P$, an approximate solution to the system given by $P$ must be close to the orbit of $\bx^*$, and furthermore one can even find such an approximate solution via a combination of exhaustive enumeration over the parameter space followed by iterative refinement.

\begin{algorithm2e}
\DontPrintSemicolon
\caption{\textsc{FindRoot}($P = (G,F),\tilde{\bc} = (\tilde{c}_0, \tilde{\bc}'),\epsilon$)}
\label{alg:findroot}
\KwIn{Polynomial map $P = (G,F):\C^N\to\C^{N+1}$, $\tilde{\bc} = (\tilde{c}_0, \tilde{\bc}')$ which is close to $\bc = P(\bx^*)$, error $\epsilon>0$}
\KwOut{$\hat{\bx}$ which is $\epsilon$-close to the $\mathcal{G}$-orbit of $\bx^*$}
    $\zeta \gets \frac{\gamma'\min(\gamma,\gamma',\gamma'',\gamma''')\,\poly_P(\delta,\smoothing)}{\calR}$\;
    $\xi_0 \gets \frac{\gamma\poly_P(\delta,\smoothing)}{\calR}$\;
    $\underline{\sigma}\gets \gamma'\poly_P(\delta,\smoothing)/\calR$ \tcp*{Lemma~\ref{lem:sigmin_lbd}}
    $T \gets \lceil\log_2(6\underline{\sigma}/\epsilon\sqrt{N+1})\rceil$\;
    Let $L$ be a $\zeta$-net over $[-\lmax,\lmax]^N$.\;
    \For(\tcp*[f]{Grid search over parameter space}){$\hat{\bx}^{(0)} \in L$}{
        \If{$\norm{F(\hat{\bx}^{(0)}) - \tilde{\bc}'}_2 \le \calR\sqrt{N}\zeta$ and $|G(\hat{\bx}^{(0)}) - \tilde{c}_0| \le \xi_0$ \label{step:checkf}}{
            Break out of loop.\; \label{step:breakout}
        }
    }
    \For(\tcp*[f]{Refine with Newton steps}){$0 \le t < T$}{
        $\hat{\bx}^{(t+1)} \gets \proj_{[-\lmax,\lmax]^N} (\hat{\bx}^{(t)} - J_P(\hat{\bx}^{(t)})^\dagger\cdot (P(\hat{\bx}^{(t)}) - \tilde{\bc}))$\; \label{step:newton}
    }
    \Return{$\hat{\bx}^{(T)}$}
\end{algorithm2e}

\begin{theorem}\label{thm:solvesystem}
    Let $\epsilon \le \gamma'\poly_P(\delta,\smoothing)/\calR$. Under Assumptions~\ref{assume:symmetry}-\ref{assume:rootsep} on the polynomial map $P: \C^N\to\C^{N+1}$, with probability $1 - \delta$ over the ground truth solution $\bx^*$, {\sc FindRoot} (Algorithm~\ref{alg:findroot}) takes as input a description of $P$ as well as a vector $\tilde{\bc} \in \C^{N+1}$ and, if 
    \begin{equation}
        \norm{\bc - \tilde{\bc}}_\infty \le \frac{\gamma'\poly_P(\delta,\smoothing)}{\calR}\min(\epsilon,\gamma,\gamma',\gamma'',\gamma''')\,, \label{eq:cclose}
    \end{equation}
    outputs $\hat{\bx}$ satisfying $\norm{\hat{\bx} - \bx^{**}}_\infty \le \epsilon$ for some point $\bx^{**}$ in the $\mathcal{G}$-orbit of $\bx^*$. The algorithm runs in time $O\Bigl(\frac{\lmax}{\gamma'\min(\gamma,\gamma',\gamma'',\gamma''')}\Bigr)^N\cdot \calR / \poly_P(\delta,\smoothing) + \poly(N) \cdot \mathcal{O}(\log 1/\epsilon)$ and succeeds with probability $1 - \delta$ over the randomness of $\bx^*$.
\end{theorem}

\begin{proof}
    We would like to apply Theorem~\ref{thm:newton} to argue that as long as $\hat{\bx}^{(0)}$ found in Line~\ref{step:breakout} is sufficiently close to a point $\bx^{**}$ in the orbit of $\bx^*$, applying enough Newton steps in Line~\ref{step:newton} will result in $\hat{\bx}^{(T)}$ which is $\epsilon$-close to $\bx^{**}$. We will take $f$ therein to be $P - \tilde{\bc}$, $\params^*$ therein to be any $\bx^{**}$ in the orbit, $\epsilon$ therein to be $\epsilon\underline{\sigma}/6\sqrt{N+1}$, $\smoothness$ therein to be $\calR N$ (by Corollary~\ref{cor:jacbounds}), and $\underline{\sigma}$ therein to $\gamma'\poly_P(\delta,\smoothing)/\calR$. Because $\norm{f(\bx^{**})}_\infty \le \norm{\bc - \tilde{\bc}}_\infty \le \epsilon\underline{\sigma}/6\sqrt{N+1}$ and $\epsilon \le \gamma'\poly_P(\delta,\smoothing)/\calR$ by assumption (the latter ensuring that Eq.~\eqref{eq:eps} is satisfied), we find that after the Newton iterations in Line~\ref{step:newton}, {\sc FindRoot} will output $\hat{\bx}^{(T)}$ which is $\epsilon$-close to $\bx^{**}$ provided that
    \begin{equation}
        \norm{\hat{\bx}^{(0)} - \bx^{**}}_\infty \le \frac{\gamma'\poly_P(\delta,\smoothing)/\calR}{2\calR N\sqrt{(N+1)N}} = \frac{\gamma'\poly_P(\delta,\smoothing)}{\calR}\,. \label{eq:warmstart}
    \end{equation}
    To achieve this, we will apply Theorem~\ref{thm:ruleout} with $\eta$ therein taken to be $\frac{\min(\gamma,\gamma',\gamma'',\gamma''')\poly_P(\delta,\smoothing)}{\calR N}$, which means that as long as we can ensure $\hat{\bx}^{(0)}$ satisfies 
    \begin{equation}
        \norm{F(\hat{\bx}^{(0)}) - \bc'}_2 \le \eta\gamma'\poly_P(\delta,\smoothing)/\calR = \frac{\gamma'\min(\gamma,\gamma',\gamma',\gamma''')\poly_P(\delta,\smoothing)}{\calR N}\,. \label{eq:xprime}
    \end{equation}
    and 
    \begin{equation}
        |G(\hat{\bx}^{(0)}) - c_0| \le \gamma\poly_P(\delta,\smoothing)/\calR\,, \label{eq:xi0}
    \end{equation}
    then the point $\hat{\bx}^{(0)}$ satisfies Eq.~\eqref{eq:warmstart} for some $\bx^{**}$ in the orbit of $\bx^*$.

    It remains to set the granularity of the net $L$ over $[-\lmax,\lmax]^N$ fine enough that we are guaranteed to find $\hat{\bx}^{(0)}$ satisfying Eqs.~\eqref{eq:xprime} and~\eqref{eq:xi0}. Note that by our choice of $\xi'$ and $\xi_0$ in Algorithm~\ref{alg:findroot}, and by the assumed bound on $\norm{\bc - \tilde{\bc}}_\infty$, any $\hat{\bx}^{(0)}$ which satisfies the condition in Line~\eqref{step:checkf} for breaking out of the loop must satisfy Eqs.~\eqref{eq:xprime} and~\eqref{eq:xi0}. 
    
    Therefore, the only thing that remains to verify is that there exists such an $\hat{x}^{(0)}$ in the net $L$. But $\norm{F(\bx^{**}) - F(\bx')}_2 \le \calR\sqrt{N}\norm{\bx^{**} - \bx'}_2$ and $|G(\bx^{**}) - G(\bx')| \le \calR\norm{\bx^{**} - \bx'}_2$ for any $\bx^{**}$ in the orbit of $\bx^*$ and any $\bx'$ in a small neighborhood of $\bx^{**}$, by the bound on $\norm{J_P}_F$ in Corollary~\ref{cor:jacbounds}. So given that $\hat{\bx}^{(0)} = \bx^{**}$ would satisfy the condition in Line~\ref{step:checkf}, any $\hat{\bx}^{(0)}$ in a $\zeta = \frac{\gamma'\min(\gamma,\gamma',\gamma'',\gamma''')\,\poly_P(\delta,\smoothing)}{\calR N^{3/2}}$-ball around $\bx^{**}$ would satisfy the condition in Line~\ref{step:checkf}. The theorem follows by our choice of $\zeta$.
\end{proof}

\begin{remark}[On verifying the assumptions]\label{remark:check}
    Apart from Assumption~\ref{assume:symmetry} which is an immediate byproduct of the problem setting, and Assumption~\ref{assume:fiber} for which we gave a prescription for how to certify in Remark~\ref{remark:certify_generic}, it may appear less clear how to check that Assumptions~\ref{assume:fiber}-\ref{assume:rootsep} for a given problem instance. It turns out, however, that each of these Assumptions immediately follows from the exact same check in Remark~\ref{remark:certify_generic}, albeit with unspecified constants.
    
    Assumptions~\ref{assume:hermite_discriminant} and~\ref{assume:rootsep} are simply quantitative versions of the condition that not all fibers of $F$ have a point with multiplicity $>1$, i.e.~a multiple zero for the polynomial $R$ defined in Lemma~\ref{lem:apply_poisson_product}. One can certify these with unspecified constants by exhibiting any fiber of the square system such that the Jacobian $J_F$ is non-singular, as is already done in the check in Remark~\ref{remark:certify_generic}. In addition, Assumption~\ref{assume:jacprod} is, by definition, certified by this check. Assumption~\ref{assume:Delta} is simply a quantitative version of generic finiteness and holds with some constant factor $\gamma''$ provided Assumption~\ref{assume:fiber} holds.
\end{remark}

\section{A General Learning Guarantee for Probe Tomography}
\label{sec:general}

Here we will fit all of the pieces from previous sections together to obtain a general learning guarantee for probe tomography of local, low-intersection Hamiltonians. The pseudocode for the associated algorithm, {\sc ProbeLearn}, can be found in Algorithm~\ref{alg:probe}. Later, in Section~\ref{sec:examples}, we will instantiate the guarantee in the Theorem below for three classes of structured Hamiltonians.

\begin{theorem}\label{thm:main_general}
    Let $H$ be a $\smoothing$-smoothed Hamiltonian from an $N$-parameter family, and denote the parameters of $H$ by $\params^*\in\R^N$.
    
    Let $S$ be a list of triples $\{(j_i,k_i,\mu_i,C_i)\}_{0\le i \le N}$, where $j_i,k_i\in \mathbb{Z}_{\ge 0}$ and $C_i$ is a single-qubit channel. Let $S'$ denote the last $N$ triples in this list. Define $\overline{j}\triangleq \max_i j_i$ and $\overline{k}\triangleq \max_i k_i$.
    
    Let $F$ (resp. $G$) denote the polynomial map $\C^N\to \C^N$ (resp. $\C^N\to \C$) that maps the parameters of the local Hamiltonian family to the corresponding values for the Taylor coefficients indexed by $S'$ (resp. $(j_0,k_0,C_0)$). Suppose $P = (G,F)$ satisfies Assumptions~\ref{assume:symmetry}-\ref{assume:rootsep} of Section~\ref{sec:quantitative} (see Remark~\ref{remark:check} on how to verify these in practice).

    Then with probability at least $1 - \delta$ over the randomness of the smoothing, {\sc ProbeLearn} (Algorithm~\ref{alg:probe}) has query complexity 
    \begin{equation}
        (\calR/\poly_P(\delta,\smoothing))\cdot \bigl[(\gamma'\min(\gamma,\gamma',\gamma'',\gamma'''))^{-2(\maxk+1)(\maxj+1)} + (1/\epsilon\gamma')^{2\maxk(\maxj+1)}\bigr]\,,
    \end{equation}
    total evolution time
    \begin{equation}
        (\calR/\poly_P(\delta,\smoothing)\cdot \bigl[(\gamma'\min(\gamma,\gamma',\gamma'',\gamma'''))^{-4\maxj-2} + (1/\epsilon\gamma')^{\maxk(2\maxj+1)}\bigr]\,,
    \end{equation}
    and classical post-processing time
    \begin{equation}
        O\Bigl(\frac{\lmax}{\gamma'\min(\gamma,\gamma',\gamma'',\gamma''')}\Bigr)^N \calR/\poly_P(\delta,\smoothing) + O\Bigr(\log \calR + \log\Bigl(\frac{1}{\epsilon\gamma'\poly_P(\delta,\smoothing)}\Bigr)\Bigr)^2\,,
    \end{equation}
    and outputs parameters $\hat{\params} \in \R^N$ for which $\norm{\params^* - \hat{\params}}_\infty\le \epsilon$. Furthermore, the temperature of the system can always be taken to be at most $\frac{\calR}{\gamma'\min(\gamma,\gamma',\gamma'',\gamma''')\,\poly_P(\delta,\smoothing)}$, in particular, independent of the final error $\epsilon$.
\end{theorem}

\noindent In particular, regarding all parameters except $\epsilon, \delta, \smoothing$ as constants (as will be the case for the examples we consider in Section~\ref{sec:examples}), we can summarize the complexity of the algorithm as follows: the query complexity and total evolution times are $\poly(1/\delta,1/\smoothing,1/\epsilon)$, and the classical post-processing time is $\poly(1/\delta,1/\smoothing,\log 1/\epsilon)$.

The reader may wonder upon seeing Theorem~\ref{thm:main_general} what more there is to prove given we have already established the general result of Theorem~\ref{thm:solvesystem} for approximate identifiability for the kinds of polynomial systems that arise in our applications. The reason is subtle: the Newton iterations in {\sc FindRoot} are ultimately bottlenecked by the error in the ``right-hand side'' of the polynomial system, as quantified by Eq.~\eqref{eq:cclose} and which comes from finite-sample error in estimating the observables and, more importantly, from truncation error when working with finite temperature. Note that in Eq.~\eqref{eq:cclose}, we need that these errors are bounded relative to the target error $\epsilon$ with which one wants to recover the ground truth solution up to symmetries. In contrast, the content of Theorem~\ref{thm:main_general} is precisely that one does not need to take inverse temperature small relative to the target error, but instead just smaller than some absolute constant depending on the Hamiltonian family that $H$ comes from.

The key step in the algorithm in this section, {\sc ProbeLearn}, is to take the solution obtained by {\sc FindRoot} and refine it further using Newton iterations. The key point is that these Newton iterations are performed not with respect to the original polynomial system $P(\bx) \approx \tilde{\bc}$, but instead with respect to a higher-degree polynomial system arising from keeping more terms --- in fact $\mathcal{O}(\log 1/\epsilon)$ many --- from the Taylor expansion of the relevant derivatives of $\obsvalue{\beta}{t}{\mu}{U}$.

\begin{proof}
    By Theorem~\ref{thm:solvesystem} with $\epsilon$ therein taken to be $\epsilon_{\sf crude} = \gamma'\,\poly_P(\delta,\smoothing)/\calR$, {\sc FindRoot} already produces an estimate $\hat{\params}$ which is $\epsilon_{\sf crude}$-close to some point $\params^{**}$ in the $\mathcal{G}$-orbit of $\params^*$, provided that the estimates $\tilde{c}_i$ are $\nu$-accurate for 
    \begin{equation}
        \nu = \frac{\gamma'\,\poly_P(\delta,\smoothing)}{\calR}\min(\gamma,\gamma',\gamma'',\gamma''')
    \end{equation}
    By Lemma~\ref{lem:derivestimate}, this can be done with probability at least $1 - \delta$ using query complexity of $$\exp(\Od(\maxj(\maxj+\maxk)))\cdot \mathcal{O}((\nu\beta^\maxk/2)^{-2\maxj-2} \log N/\delta) = (\calR/\poly_P(\delta,\smoothing)\cdot (\gamma'\min(\gamma,\gamma',\gamma'',\gamma'''))^{-2(\maxk+1)(\maxj+1)},$$
    and total evolution time
    $$\exp(\Od(\maxj(\maxj+\maxk)))\cdot \mathcal{O}((\nu\beta/2)^{-2\maxj-1}  \log N/\delta) = (\calR/\poly_P(\delta,\smoothing))\cdot (\gamma'\min(\gamma,\gamma',\gamma'',\gamma'''))^{-4\maxj-2},$$ provided that $\beta \le \nu \,\Omega_{\mathfrak{d}}(\maxk)^{-\maxk-1} \exp(-\Omega_{\mathfrak{d}}(\maxj)) = \frac{\gamma'\,\poly_P(\delta,\smoothing)}{\calR}\min(\gamma,\gamma',\gamma'',\gamma''')$. Both query complexity and total evolution time scale as $\poly(1/\varepsilon)$.

    Next, we run Newton's method with the following choice of differentiable map $f: \R^N\to\R^{N+1}$. First, define $\overline{k}_i\triangleq k_i + \Theta_{\mathfrak{d}}(\frac{\log(1/\epsilon') + j + k}{\log 1/\beta})$ for $\epsilon' > 0$ to be tuned later. For every $0\le i \le N$, instead of the polynomial $\obsvalue{\beta^{(k_i)}}{t^{(j_i)}}{\mu_i}{C_i}$, define $f_i(\params) \triangleq \hideg{j_i}{k_i}{\overline{k}_i}(\params)$, where $\hideg{j}{k}{\overline{k}}$ is the polynomial defined in Eq.~\eqref{eq:extraterms} but where we have abused notation and added $(\params)$ to indicate that we are regarding it as a formal polynomial in the unknown parameters of the Hamiltonian. Note from Eq.~\eqref{eq:extraterms} that the lowest-order term in $f_i(\params)$ is exactly $\obsvalue{\beta^{(k_i)}}{t^{(j_i)}}{\mu_i}{C_i}$, but we are including roughly $\mathcal{O}(\log 1/\epsilon)$ additional terms to lessen the truncation error.

    We wish to apply Theorem~\ref{thm:newton} to the map $f$. The key step will be to bound the smoothness $\smoothness$ and curvature $\underline{\sigma}$ of $f$, at which point we can take $\epsilon$ therein to be $\Theta(\min(\epsilon,\underline{\sigma}/\sqrt{N})\cdot \underline{\sigma}/\sqrt{N})$ and conclude that if $\epsilon_{\sf crude} \le \mathcal{O}(\underline{\sigma}/\smoothness N)$, then in $\mathcal{O}(\log(\underline{\sigma}/\epsilon))$ Newton steps we end up at a point which is $\epsilon$-close to whichever point $\params^{**}$ in the orbit of $\bx^*$ we initialized at.

    To compute $\underline{\sigma}$, recall from Lemma~\ref{lem:sigmin_lbd} that for the original polynomial map $P = (G,F)$ that we ran {\sc FindRoot} on in Line~\ref{step:findroot} of {\sc ProbeLearn}, $\sigma_{\min}(J_P(\params^{**})) \ge \gamma'\,\poly_P(\delta,\smoothing)/\calR$ with probability at least $1 - \delta$ over the randomness of the smoothed Hamiltonian. Recall from Eq.~\eqref{eq:extraterms} that
    \begin{equation}
        \hideg{j_i}{k_i}{\overline{k}_i} = P_i + \sum_{k_i<k'\le \overline{k}_i} \beta^{k'-k_i}\binom{k'}{k_i} \sum_a c_a p_{a;k'}\,,
    \end{equation}
    where $c_a, p_{a;k'}$ depend on $j_i, \mu_i, C_i$. The gradient of $\hideg{j_i}{k_i}{\overline{k}_i} - P_i$ with respect to $\params$ is thus
    \begin{equation}
        \sum_{k_i<k'\le \overline{k}_i} \beta^{k'-k_i}\binom{k'}{k_i} \sum_a c_a \nabla_{\params}\,p_{a;k'}\,. \label{eq:grad_diff}
    \end{equation}
    Because the $c_a$'s are $\exp(\mathcal{O}(j))$-bounded, there are $\exp(\Od(j))$ many of them, and the partial derivatives of $p_{a;k'}$ are $\exp(\Od(k'))$-bounded, the $L_2$ norm of the above is at most
    \begin{equation}
        \sqrt{N}\exp(\Od(j))\beta^{-k_i} \sum_{k' > k_i} \Od(\beta)^{k'} \le \sqrt{N}\exp(\Od(j))\cdot \beta\,,
    \end{equation}
    provided $\beta \le \beta_{\sf crit}$. If furthermore $\beta \le \gamma'\,\poly_P(\delta,\smoothing)/\calR$, then by Weyl's inequality, we conclude that we can take $\underline{\sigma}$ to be
    \begin{equation}
        \sigma_{\min}(J_f) \ge \gamma'\,\poly_P(\delta,\smoothing)/\calR - \sqrt{N}\exp(\Od(j))\cdot \beta \ge \gamma'\,\poly_P(\delta,\smoothing)/\calR\,.
    \end{equation}

    Next, we compute $\smoothness$. Recall from Corollary~\ref{cor:jacbounds} that the Hessian of $P_j$ for each $j\in[N+1]$ is upper bounded in norm by $\calR$. The Hessian of $f_j$ is at most this plus the operator norm of 
    \begin{equation}
        \sum_{k_i<k'\le \overline{k}_i} \beta^{k'-k_i}\binom{k'}{k_i} \sum_a c_a \nabla^2_{\params}\,p_{a;k'}\,. \label{eq:hess_diff}
    \end{equation}
    and by an argument identical to the above, the operator norm of the latter is dominated by that of the Hessian of $P_j$ for the range of $\beta$ we are considering. Substituting these bounds into Theorem~\ref{thm:newton}, we conclude that because $\epsilon_{\sf crude} \le \gamma'\,\poly_P(\delta,\smoothing)/\calR \le \mathcal{O}(\underline{\sigma}/\smoothness N)$, provided that we were able to estimate each $\hideg{j_i}{k_i}{\overline{k}_i}$ to error $\epsilon\gamma'\,\poly_P(\delta,\smoothing)/\calR$, then in $\mathcal{O}(\log \calR + \log(\frac{1}{\epsilon\gamma'\poly_P(\delta,\smoothing)}))$ Newton steps we converge to an $\epsilon$-accurate solution. Estimating $\hideg{j_i}{k_i}{\overline{k}_i}$ to this level of error requires taking $\epsilon'$ in the definition of $\overline{k}_i$ and $h$ in Lemma~\ref{lem:hideg} to be $\epsilon\gamma'\,\poly_P(\delta,\smoothing)/\calR$, which translates to query complexity $(\calR/\poly_P(\delta,\smoothing))\cdot (1/\epsilon\gamma')^{2\maxk(\maxj+1)}$ and total evolution time $(\calR/\poly_P(\delta,\smoothing))\cdot (1/\epsilon\gamma')^{\maxk(2\maxj+1)}$.

    The classical post-processing time includes the time for {\sc FindRoot} which is dominated by the grid search which scales exponentially in the number of parameters $N$, and the time for setting up the nonlinear system given by the $f_i$'s, and the time for running Newton's method. By the bound in the last part of Lemma~\ref{lem:hideg}, the runtime for writing down the $f_i$'s is of lower order relative to the cost of Newton's method. For the cost of Newton's method, there are two log factors appearing in the final bound: the first log factors corresponds to the number of terms that appear in each polynomial $\hideg{j}{k}{\overline{k}}$ and in particular our bound on $\overline{k} - k$, and the second corresponds to the number of Newton steps.
\end{proof}

\begin{algorithm2e}
\DontPrintSemicolon
\caption{\textsc{ProbeLearn}($S, H, \epsilon$)}
\label{alg:probe}
\KwIn{$S = \{(j_i,k_i,\mu_i,C_i)\}_{0\le i \le N}\subset \mathbb{Z}_{\ge 0}\times\mathbb{Z}_{\ge 0}\times \{\text{1-qubit channels}\}$, probe access to $H$ with unknown parameters $\params^*$, accuracy parameter $\epsilon>0$}
\KwOut{$\hat{\params}$ which is $\epsilon$-close to the $\mathcal{G}$-orbit of $\params^*$}
    $\epsilon_{\sf crude}\gets \gamma'\,\poly_P(\delta,\smoothing)/\calR$\;
    $\epsilon'\gets \epsilon\gamma'\,\poly_P(\delta,\smoothing)/\calR$\;
    \For{$0\le i \le N$}{
        Form estimate $\tilde{c}_i$ for $\obsvalue{\beta^{(k_i)}}{t^{(j_i)}}{\mu_i}{C_i}$ using $(\calR/\poly_P(\delta,\smoothing)\cdot (\gamma'\min(\gamma,\gamma',\gamma'',\gamma'''))^{-2(\maxk+1)(\maxj+1)}$ queries and $(\calR/\poly_P(\delta,\smoothing)\cdot (\gamma'\min(\gamma,\gamma',\gamma'',\gamma'''))^{-4\maxj-2}$ total evolution time.\tcp*{Lemma~\ref{lem:derivestimate}}
    }
    $\hat{\params}^{(0)} \gets ${\sc FindRoot}($P, \tilde{\bc}, \epsilon_{\sf crude}$) \label{step:findroot} \tcp*{Obtain initial estimate}
    $T\gets \mathcal{O}(\log\calR + \log(\frac{1}{\epsilon\gamma'\,\poly_P(\delta,\smoothing)}))$\;
    $\beta\gets \min(\beta_{\sf crit}, \gamma'\,\poly_P(\delta,\smoothing)/\calR)$\;
    \For(\tcp*[f]{Form higher-degree system}){$0\le i \le N$}{
        $\overline{k}_i \gets k_i+ \Theta_{\mathfrak{d}}(\frac{\log(1/\epsilon') + j + k}{\log 1/\beta})$\;
        Form $\epsilon\gamma'\,\poly_P(\delta,\smoothing)/\calR$-accurate estimate $\tsup{c}_i$ for $\hideg{j_i}{k_i}{\overline{k}_i}$ using $(\calR/\poly_P(\delta,\smoothing))\cdot (1/\epsilon\gamma')^{2\maxk(\maxj+1)}$ queries and total evolution time $(\calR/\poly_P(\delta,\smoothing))\cdot (1/\epsilon\gamma')^{\maxk(2\maxj+1)}$. \tcp*{Lemma~\ref{lem:hideg}}
        Let $f_i: \R^N\to\R$ denote the polynomial  $\params \mapsto \hideg{j_i}{k_i}{\overline{k}_i}(\params) - \tsup{c}_i$.\;
    }
    \For(\tcp*[f]{Refine with Newton steps}){$0 \le t < T$}{
        $\hat{\params}^{(t+1)} \gets \proj_{[-\lmax,\lmax]^N} \bigl[\hat{\params}^{(t)} - J_{f}(\hat{\params}^{(t)})^\dagger\cdot f(\hat{\params}^{(t)})\bigr]$\; \label{step:newton2}
    }
    \Return{$\hat{\params}^{(T)}$}
\end{algorithm2e}

\section{Application to Nearest-Neighbor Hamiltonians}
\label{sec:examples}

Here we will consider quantum probe tomography for nearest-neighbor, translation-invariant, isotropic Hamiltonians in $D = 1, 2, 3$ spatial dimensions on a square lattice.  We write $G = (V,E)$ for the graph of the lattice where $V = \mathbb{Z}^D$ and $E$ are the nearest-neighbor edges.  This  class of Hamiltonians can be parameterized as
\begin{align}
H = \sum_{\langle v, v'\rangle \in E}\,\sum_{\mu, \nu =1}^3 J_{\mu \nu}\, \sigma^\mu_v \sigma^\nu_{v'} + \sum_{v \in V} \sum_{\mu = 1}^3 \,h_\mu\sigma^\mu_v\,,
\end{align}
where $\sigma^\mu_v$ is that $\mu$th Pauli operator on site $v$, and so evidently has $12$ parameters
\begin{align}
(h_1, h_2, h_3, J_{11}, J_{12}, J_{13}, J_{21}, J_{22}, J_{23}, J_{31}, J_{32}, J_{33})
\end{align}
for any $D$.  We will truncate the Hamiltonian so that the sites extend to a large but finite volume centered around $\probeindex \in \mathbb{Z}^D$ (possibly with periodic or hard-wall boundary conditions), so that the total Hilbert space dimension is $d$.  Suppose we only have access to probing the central site $\probeindex \in \mathbb{Z}^D$.  As a shorthand, we will sometimes denote the Pauli operators $\sigma^{1}_\probeindex, \sigma^{2}_\probeindex, \sigma^{3}_\probeindex$ on this site by, simply, $X, Y, Z$.  Let us also define the unitary channels
\begin{align}
C_0[\rho] &= \rho\,,\qquad C_1[\rho] = X \rho X\,,\qquad \, C_2[\rho] = Y \rho Y\,, \qquad \, C_3[\rho] = Z \rho Z\,, \\
C_4[\rho] &= \frac{1}{\sqrt{2}}(X + Y) \rho \frac{1}{\sqrt{2}}(X + Y)\,,\qquad C_5[\rho] = \frac{1}{\sqrt{2}}(Y + Z) \rho \frac{1}{\sqrt{2}}(Y + Z)\,, \\
C_6[\rho] &= \frac{1}{\sqrt{2}}(Z + X) \rho \frac{1}{\sqrt{2}}(Z + X)\,, \qquad C_7[\rho] = \frac{1}{\sqrt{2}}(\mathds{1} + \i X) \rho \frac{1}{\sqrt{2}}(\mathds{1} - \i X)\,, \\
C_8[\rho] &= \frac{1}{\sqrt{2}}(\mathds{1} + \i Y) \rho \frac{1}{\sqrt{2}}(\mathds{1} - \i Y)\,, \quad \,\,\,\,\, C_9[\rho] = \frac{1}{\sqrt{2}}(\mathds{1} + \i Z) \rho \frac{1}{\sqrt{2}}(\mathds{1} - \i Z)\,.
\end{align}
For our purposes it will suffice to consider interventions at one, two, and three separate times.  To this end, we define $\rho_\beta = e^{-\beta H}/\tr(e^{- \beta H})$ and the notation
\begin{align}
\obsvalue{\beta}{t}{\mu}{B} \triangleq \tr(\sigma^\mu_{\probeindex} (C_B[\rho_\beta])_H(t))\,,\qquad \mu = 0,1,2,3,\quad B = 0,...,9\,,
\end{align}
which we have borrowed from~\eqref{E:observablesA0}.  We recall that $\obsvalue{\beta^{(k)}}{t^{(j)}}{\mu}{B}$ denoted the coefficient of $t^j\beta^k$ in the Taylor series expansion of $\obsvalue{\beta}{t}{\mu}{C}$.
Given that $\tr(H) = 0$, we have
$\rho_\beta = \frac{\mathds{1}}{d} - \frac{\beta}{d}\, H + \frac{\beta^2}{2d}\left(H^2 - \tr(H^2)\, \frac{\mathds{1}}{d}\right) + \mathcal{O}(\beta^3)$.
We construct the polynomials
\begin{align}
p_1 &= - \obsvalue{\beta^{(1)}}{t^{(0)}}{0}{1} = \frac{1}{d}\,\tr(X H) \\
p_2 &= - \obsvalue{\beta^{(1)}}{t^{(0)}}{0}{2} = \frac{1}{d}\,\tr(Y H) \\
p_3 &= - \obsvalue{\beta^{(1)}}{t^{(0)}}{0}{3} = \frac{1}{d}\,\tr(Z H) \\
p_4 &= \obsvalue{\beta^{(2)}}{t^{(0)}}{0}{1} = \frac{1}{2d}\,\tr(X H^2) \\
p_5 &= \obsvalue{\beta^{(2)}}{t^{(0)}}{0}{2} = \frac{1}{2d}\,\tr(Y H^2) \\
p_6 &= \obsvalue{\beta^{(2)}}{t^{(0)}}{0}{3} = \frac{1}{2d}\,\tr(Z H^2) \\
p_7 &= -\frac{1}{4} \obsvalue{\beta^{(1)}}{t^{(1)}}{1}{2} = \frac{1}{4 i d}\,\tr([H,X] Y H Y) \\
p_8 &= -\frac{1}{4} \obsvalue{\beta^{(1)}}{t^{(1)}}{2}{3} = \frac{1}{4i  d}\,\tr([H,Y] Z H Z)  \\
p_9 &= -\frac{1}{4} \obsvalue{\beta^{(1)}}{t^{(1)}}{3}{1} = \frac{1}{4 i d}\,\tr([H,Z] X H X) \\
p_{10} &= -\frac{1}{4}\big(\obsvalue{\beta^{(1)}}{t^{(1)}}{3}{9} - \obsvalue{\beta^{(1)}}{t^{(1)}}{3}{4}\big) = \frac{1}{4id}\,\tr([H,Z] (C_9 - C_4)[H]) \\
p_{11} &=-\frac{1}{4}\big(\obsvalue{\beta^{(1)}}{t^{(1)}}{2}{7} - \obsvalue{\beta^{(1)}}{t^{(1)}}{2}{5}\big) = \frac{1}{4id}\,\tr([H,X](C_7-C_5)[H]) \\
p_{12} &= -\frac{1}{4}\big(\obsvalue{\beta^{(1)}}{t^{(1)}}{2}{8} - \obsvalue{\beta^{(1)}}{t^{(1)}}{2}{6}\big) = \frac{1}{4id}\,\tr([H,Y](C_8-C_6)[H]) \\
q &= \frac{1}{4} \big(\obsvalue{\beta^{(1)}}{t^{(2)}}{1}{1}  + \obsvalue{\beta^{(1)}}{t^{(2)}}{1}{2}\big)= -\frac{1}{8d}\,\tr\big([H,X]\,[H, XHX + YHY]\big)\,,
\end{align}
which have the explicit form
\begin{align}
p_1 &= h_1 \\
p_2 &= h_2 \\
p_3 &= h_3 \\
p_4 &= D\,h_1 (2J_{11} + J_{12} + J_{13} + J_{21} + J_{31}) \\
p_5 &= D\,h_2 (J_{12} + J_{21} + 2J_{22} + J_{23} + J_{32}) \\
p_6 &= D\,h_3 (J_{13} + J_{23} + J_{31} + J_{32} + 2J_{33}) \\
p_7 &= h_2 h_3 + D\,(J_{12}J_{13} + J_{22}J_{23} + J_{21}J_{31} + J_{22}J_{32} + J_{23}J_{33} + J_{32}J_{33}) \\
p_8 &= h_1 h_3 + D\,(J_{11}J_{13} + J_{21}J_{23} + J_{11}J_{31} + J_{12}J_{32} + J_{13}J_{33} + J_{31}J_{33}) \\
p_9 &= h_1 h_2 + D\,(J_{11}J_{12} + J_{11}J_{21} + J_{12}J_{22} + J_{21}J_{22} + J_{13}J_{23} + J_{31}J_{32}) \\
p_{10} &= h_1^2 + D\,(2 J_{11}^2 + J_{12}^2 + J_{21}^2 +  J_{13}^2 + J_{31}^2) \\
p_{11} &= h_2^2 + D\,(J_{21}^2 + J_{12}^2 + 2 J_{22}^2 +  J_{23}^2 + J_{32}^2) \\
p_{12} &= h_3^2 + D\,(J_{31}^2 + J_{13}^2 + J_{23}^2 + J_{32}^2 +  2 J_{33}^2) \\
q &= h_1 h_3^2 
+ D\Big(3 h_2 \left(-J_{23} J_{31} - J_{13} J_{32} + (J_{12} + J_{21}) J_{33}\right) \nonumber \\
& \qquad \qquad \qquad + h_1 \left(J_{13}^2 + J_{23}^2 + J_{31}^2 + J_{32}^2 + 6 J_{23} J_{32} + 2 J_{33} \left(J_{33} - 3 J_{22}\right)\right) 
\nonumber \\
& \qquad \qquad \qquad + h_3 \left(J_{21} \left(2 J_{23} - 3 J_{32}\right) + J_{12} \left(2 J_{32} - 3 J_{23}\right) + (J_{13} + J_{31}) \left(2 J_{11} + 3 J_{22} + 2 J_{33}\right)\right)\Big).
\end{align}
We note that the first twelve polynomials are at most quadratic; all other at most quadratic polynomials expressible in terms of the $A$'s are linear combinations of these twelve.

Notice that our polynomials are all invariant under $J_{\mu \nu} \to J_{\nu \mu}$\,; this is because (to the order in perturbation theory we are working) the measurements corresponding to the polynomials are invariant under an inversion of the lattice about the point $\probeindex$.

Let us package the thirteen polynomials as $P = (q, p_1,...,p_{12})$.  By plugging in a particular (rational) value of the couplings $\textbf{x}_0$ and exactly  solving the polynomial system $P(\textbf{x}) = P(\textbf{x}_0) =: P(\textbf{c}_0)$ to obtain the inverse image, we find two points, which we verify each have full rank Jacobian.  This verifies that the generic fiber size is two, and therefore the second of the two points in a generic fiber is the inversion of the `true' Hamiltonian, namely with $J_{\mu \nu} \to J_{\nu \mu}$.  As such, we can recover the `true' Hamiltonian up to this inversion symmetry.  Moreover, for the square system $(p_1,...,p_{12})$, we run the same procedure and likewise verify a finite generic fiber size where each fiber generically has a full rank Jacobian.

Taken together, the above checks establish that our system of polynomials satisfies all of our listed assumptions for our quantum probe tomography algorithm.  When instantiating Theorem~\ref{thm:main_general} in this setting, all relevant parameters are constant except $\delta, \smoothing, \epsilon$, and we obtain the guarantee claimed in Theorem~\ref{thm:application_intro}.

\bibliographystyle{alpha}
\bibliography{ref}

\end{document}